\documentclass[3p, preprint, review]{elsarticle}

\usepackage{graphics, graphicx}
\usepackage{amsthm}
\usepackage{amsmath}
\usepackage{amssymb, amsfonts}
\usepackage{subcaption, caption}

\newcommand{\Jcal}{\mathcal{J}}
\newcommand{\Pcal}{\mathcal{P}}
\newcommand{\Ycal}{\mathcal{Y}}

\newcommand{\Ebb}{\mathbb{E}}
\newcommand{\Pbb}{\mathbb{P}}
\newcommand{\Rbb}{\mathbb{R}}

\newcommand{\Cb}{\mathcal{C}_b}
\newcommand{\Cp}{\mathcal{C}_p}
\newcommand{\Cd}{\mathcal{C}_d}

\newcommand{\Ind}[1]{\mathbb{I}_{\{#1\}}}
\newcommand{\abs}[1]{\left|#1\right|}
\newcommand{\sgn}{\text{sgn}}
\DeclareMathOperator*{\argmin}{argmin}

\newtheorem{definition}{Definition}
\newtheorem{proposition}{Proposition}
\newtheorem{theorem}{Theorem}

\usepackage{color}

\begin{document}

\begin{frontmatter}
\title{Power Control for Packet Streaming with Head-of-Line
  Deadlines}

\author{Neal~Master\corref{master}\fnref{master}} 
\ead{nmaster@stanford.edu} 
\address{Department of Electrical Engineering, Stanford University,
  Stanford CA, 94305} \cortext[master]{Corresponding author.}
\author{Nicholas~Bambos\corref{}}
\ead{bambos@stanford.edu}
\address{Department of Electrical Engineering and Department of
  Management Sciences \& Engineering, Stanford University, Stanford
  CA, 94305}
\fntext[master]{Neal Master is funded by Stanford University through a
  Stanford Graduate Fellowship (SGF) in Science \& Engineering.}

\begin{abstract}
  We consider a mathematical model for streaming media packets (as the
  motivating key example) from a transmitter buffer to a receiver over
  a wireless link while controlling the transmitter power (hence, the
  packet/job processing rate). When each packet comes to the
  head-of-line (HOL) in the buffer, it is given a deadline $D$ which
  is the maximum number of times the transmitter can attempt
  retransmission in order to successfully transmit the packet. If this
  number of transmission attempts is exhausted, the packet is ejected
  from the buffer and the next packet comes to the HOL. Costs are
  incurred in each time slot for holding packets in the buffer,
  expending transmitter power, and ejecting packets which exceed their
  deadlines. We investigate how transmission power should be chosen so
  as to minimize the total cost of transmitting the items in the
  buffer. We formulate the optimal power control problem in a dynamic
  programming framework and then hone in on the special case of fixed
  interference. For this special case, we are able to provide a
  precise analytic characterization of how the power control should
  vary with the backlog and how the power control should react to
  approaching deadlines. In particular, we show monotonicity results
  for how the transmitter should adapt power levels to the backlog and
  approaching deadlines. We leverage these analytic results from the
  special case to build a power control scheme for the general
  case. Monte Carlo simulations are used to evaluate the performance
  of the resulting power control scheme as compared to the optimal
  scheme. The resulting power control scheme is sub-optimal but it
  provides a low-complexity approximation of the optimal power
  control. Simulations show that our proposed schemes outperform
  benchmark algorithms. We also discuss applications of the model to
  other practical operational scenarios.
\end{abstract}

\end{frontmatter}

\section{Introduction} 

Packet streaming over wireless channels is a ubiquitous technology
today.  Indeed, the rise of smartphones, tablets, wearable
electronics, and the Internet of Thing (IoT) has led to an increasing
interest in wireless multimedia applications. However, real-time
mobile multimedia streaming poses challenges in both theory and
practice. Mobile devices have strict power limitations, and as they
become more compact, energy efficiency becomes increasingly
important. In addition, multimedia streaming is time-sensitive with
both latency and jitter constraints. These constraints are complicated
by the fact that wireless channel quality fluctuates stochastically in
both time and space. These limitations and constraints of a wireless
streaming system are typically at odds with one another which makes
practical yet effective control schemes difficult to formulate.

In the literature, wireless streaming over cellular networks has been
studied in the context of downlink packet scheduling
\cite{Fattah_WirelessScheduling_2002} \cite{Dua_CD2_2010}.  In such
systems, a base station performs the task of transmitting multimedia
streams over various wireless channels to different end users. Since
the base station has limited resources, there is a need for scheduling
algorithms for supporting the multitude of streams. In this case,
channel quality and time-sensitivity require consideration but the
base station does not have the same stringent power restrictions as a
mobile device.

A different model is the case of device-to-device (D2D) transmission
\cite{Asadi_D2D_2013}. In these systems, devices communicate directly
over local wireless channels without intermediate base stations or
access points. D2D networks have been implemented and studied in the
context of LTE \cite{Zulhasnine_D2DLTE_2010} as well as Wi-Fi
\cite{Camps-Mu_D2DWifi_2013}. D2D communication networks are flexible
and can be used for applications as diverse as content distribution
and vehicular communication. As a result, a subset of a D2D network
could likely be a device streaming video to another device. For
example, if a particular user has already downloaded a video,
streaming the content directly to a nearby user avoids complexity of
downlink scheduling and hence reduces the system of interest to a
point-to-point communication model. In this paper we consider such a
point-to-point system and develop schemes which incorporate power
constraints, varying wireless channel quality, as well as latency and
jitter constraints. Rather than focusing on a specific set of
technologies, we use an abstract mathematical model and prove
structural results about the optimal transmitter power control. These
properties are used to develop low-complexity schemes and the design
principles of these schemes are demonstrated via simulation. The
abstract nature of the model allows for the principles to be applied
to a variety of technologies.

\subsection{Related Work}
There has been considerable work on the case of streaming multimedia
over a wireless link. Li et al. \cite{Li_PowerPlayout_2006}
formulated the problem as a joint transmitter power and receiver
playout rate control problem. They used the optimal control to
motivate useful heuristics for jointly choosing the transmitter power
and the receiver playout rate when communicating over an
interference-limited wireless link. These heuristics can be used to
ensure ``smooth'' multimedia playout but they do not explicitly
capture jitter constraints. In particular, the heuristics do not make
use of any form of packet deadlines.

Xing et al. \cite{Xing_D2DWifi_2009} conducted an experimental study
of D2D video delivery using Wi-Fi ad-hoc mode. This work gives
empirical insights into the technological details of multimedia
streaming over Wi-Fi. Unfortunately, power control is not yet part of
the IEEE 802.11 standard (which is the basis for Wi-Fi)
\cite{IEEE_80211}, so the study was unable to incorporate any
investigation of transmitter power control.

There has been substantial work on power allocation at the network
level; a thorough survey is presented by Chiang et
al. \cite{Chiang_Power_2008}.  Le \cite{Le_D2DResources_2012} studied
fair resource allocation in D2D Orthogonal Frequency Division Multiple
Access (OFDMA)-based wireless cellular networks. Power was included as
a resource and the focus was on general traffic rather than
specifically multimedia. Yu et al. \cite{Yu_D2DResources_2011} studied
how to share resources between a D2D network and traditional
underlying cellular network. In a general setting (not necessarily
D2D), Kandukuri and Boyd \cite{Kandukuri_Outage_2002} used a convex
programming framework to optimally allocate power across a network so
as to minimize outage probabilities. These works address the question
of how to assign a power budget to each transmitter in a network of
transmitters. Our work builds on this idea by investigating how a
transmitter should dynamically choose to use the allocated power
budget.

Dynamic power control and delay control have been considered in the
network utility maximization literature. For example, O'Neill et
al. \cite{ONeill_NUM_2008} used a convex optimization framework to
derive power and rate control policies for several modulation
schemes. Although rate is implicitly related to the delay of
information, this work did not explicitly consider deadlines or delay
bounds. In another line of research, Neely \cite{Neely_Delay_2013}
used Lyapunov based techniques to approximately maximize throughput
utility based on explicit head-of-line packet delays. Neely
\cite{Neely_EnergyDelay_2007} also investigated the asymptotic
trade-off between energy and delay in a wireless downlink system. Our
work is different not only in the approach (we use a dynamic
programming framework) but also in that our results are not
asymptotic.

Outside of the communication engineering community, queueing models
with deadlines have been considered in the scheduling
literature. Delay-sensitive scheduling has applications in networking
of computers \cite{Kim_Broadcasts_2004} and patient triage
\cite{Argon_Scheduling_2008}. More recently, there has been work on
scheduling with ``impatient'' users \cite{Dalal_Impatient_2005}. In
such a model, jobs do not have hard deadlines but users have
preferences which create soft deadlines. In all of these works, the
focus is on scheduling jobs under different notions of delay
sensitivity where each server completes jobs with a fixed rate. In the
current work, we are concerned with a complementary issue. The
schedule for the packets is fixed (they are served sequentially) but
we control the service rate (via the transmitter power) and are
interested in understanding how the service rate should vary with the
number of jobs (packets) and the approaching deadlines. Service rate
control for a single Markovian queue has been previously considered
\cite{George_ServiceRateControl_2001}.  Besides the focus on wireless
streaming, our model differs in a few ways with the most notable
difference being the head-of-line deadlines.

Early limited results of this work were presented in our previous
conference publications \cite{Master_HOL_2014, Master_ACC_2015}, where
a baseline model was developed and and some preliminary analysis
done. The current paper presents a detailed mathematical treatment of
the full system model (including random interference), resulting
control schemes and emerging design principles, and incorporates
additional key metrics into the performance evaluation.

\subsection{Contributions}
Our goal in this paper is to study a system in which a transmitter is
streaming a sequence of multimedia packets over a time-varying
wireless link. Because of the aforementioned time-sensitivity of the
data, there is a bound on the number of transmission attempts for each
packet, i.e. there are head-of-line (HOL) packet deadlines. We
formulate an optimal control problem for the transmitter to choose the
optimal transmit power level as a function of the residual HOL
deadline and the remaining number of packets. We provide tight
theoretical results on the monotoniticy properties of the optimal
policy. In particular, we analytically characterize 1) how the
transmitter power should vary with the backlog and 2) how the
transmitter should react to approaching deadlines. Based on these
structural results, we next develop 3) low-complexity power control
schemes, which approximate the optimal control efficiently. We
evaluate the performance of these schemes via simulation. The abstract
mathematical framework of this work allows for these design principles
and approximation techniques to be applied to a variety of
technologies.

\subsection{Extended Model Applications}
The model developed in Section~\ref{sec:model} actually maps to
general situations of 1) {\em streamed/sequential} processing of jobs
with 2) {\em head-of-line (HOL) deadlines}, possibly in the presence
of a 3) {\em fluctuating environment}. The issue is to control the
{\em job processing rate} to minimize the overall processing cost,
comprised of a 1) {\em rate cost} (higher rate, higher cost), a 2)
{\em backlog/holding cost} for keeping unprocessed jobs in the buffer,
and a 3) HOL {\em missed-deadline cost} incurred when the HOL job
misses its deadline and, thus, has to be dropped.

In this paper, we opt to project this general model to wireless media
streaming, because this is a key modern application. In this case, the
media packets are the jobs, the wireless channel is the fluctuating
environment, and the HOL deadline is the time by which the media
packet under transmission should successfully reach the receiver
(otherwise it is dropped and the next one advances to the HOL). The
HOL packet processing rate is induced by the selected transmission
power (hence, the terminology {\em power control}) in the media
streaming application.

The model, however, has other interesting applications too, in
particular for tasks with time-outs. For example, a high-level view of
(virtualized) {\em data centers} -- which captures the trade-off
between computation task latency, time-out rate, and computation
resource cost -- is the following. Computation tasks (of a given class
X) are queued up for FIFO processing.  When a task starts execution,
it will {\em time out} and be dropped, if it does not complete before
a time-out horizon (deadline). For example, such is the case for
certain database tasks (e.g., related to financial/e-commerce
transactions with finite active exposure horizons, due to security
reasons) or computation tasks which access time-sensitive data. The
HOL task of class X completes service in the current time slot with a
certain probability (service rate), which depends on 1) the
computation resources allocated to the class X task stream and 2) the
congestion level of the overall computation resources (random
environment), which serve various other classes (e.g., Y, Z, etc.) The
more resources are allocated to class X, the lower its task latency
and time-out rate. However, the cost of allocating more resources to
class X gets higher as the overall (fluctuating) resource environment
becomes more congested. The issue is to balance the cost of resources
allocated to class X against its task latency and time-out rate.

\subsection{Paper Outline}
The remainder of this paper is organized as follows: In Section
\ref{sec:model} we describe our system model including the wireless
channel, the packet queueing dynamics, and the associated costs.  We
use this model to formulate an optimal control problem. We explain the
practical difficulties of the original optimal control problem and, in
Section~\ref{sec:illustrative}, focus on an illustrative case of the
general problem. We provide numerical and analytical results for this
special case. In particular, we are able to write the optimal policy
in a semi-analytic form which does not involve the optimal cost-to-go
function. This allows us to prove the aforementioned structural
properties of the optimal policy. In Section~\ref{sec:SLBPC}, we
leverage the analytic results from the previous section to design a
low-complexity power control scheme called Sublinear-Backlog Power Control
(SLBPC) which does not require numerically solving the original optimal
control problem. In Section~\ref{sec:performance}, we demonstrate via
simulation that this low-complexity scheme is a good approximation of
the optimal power control. We conclude in Section
\ref{sec:conclusions}.

\section{System Model and Optimal Control\label{sec:model}}
In this section we mathematically model a discrete-time point-to-point
wireless communication system. We begin by presenting a standard model
for a stochastically fluctuating wireless channel. We then present the
system dynamics and costs for our streaming system. The overall model
is shown diagramatically in Figure~\ref{fig:model}. The model is used
to formulate our power control problem in an optimal control
framework. Finally, we discuss the practical difficulties associated
with an optimal control approach.

\begin{figure}[ht]
  \centering
  \includegraphics[width=0.8\textwidth]{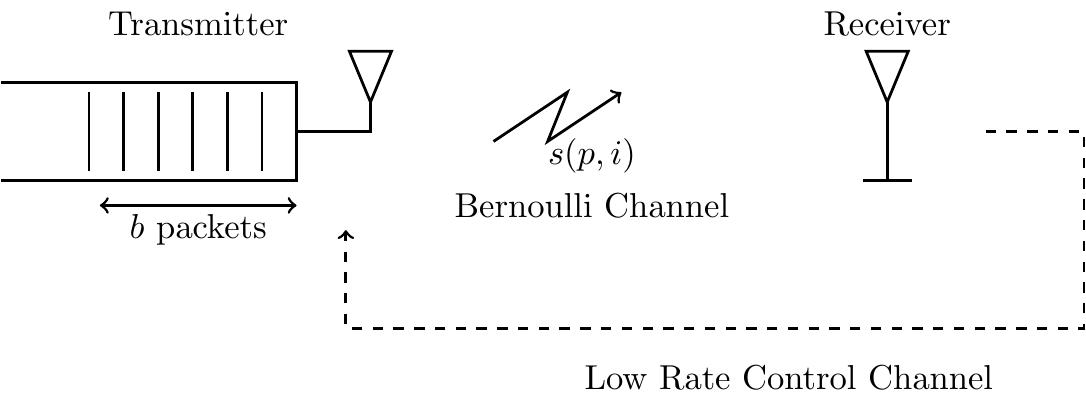}
  \caption{ A diagram of the system model. The transmitter has a
    buffer of packets which are sequentially transmitted over a
    wireless channel. If the transmitter uses power $p$ and the
    interference is $i$, the probability of successfully transmitting
    the packet is $s(p, i)$. In this case, the receiver sends an
    acknowledgment back to the transmitter over a low-rate control
    channel. With probability $1 - s(p, i)$, the transmission is
    unsuccessful. In this case, the transmitter will attempt to
    re-transmit in the next time slot. Each packet has a deadline so
    that the transmitter only has $D$ transmission attempts per
    packet. If these attempts are exhausted, the packet is ejected
    from the buffer. A more complete description of the system model
    can be found in Section~\ref{ssec:channel} and in
    Section~\ref{ssec:dynamics}. \label{fig:model}}
\end{figure}

\subsection{Fluctuating Wireless Channel \label{ssec:channel}}
In each time slot, if the transmitter uses power $p$ and the channel
has interference state $i$, the probability of successful transmission
is $s(p, i)$. We assume that each channel use is statistically
independent of every other channel use. The function $s(\cdot, \cdot)$
is assumed to have the following properties:
\begin{enumerate}[(i)]
\item For any fixed $i$, $p \mapsto s(p, i)$ is non-decreasing.
\item For any $(p, i)$, $s(p, i) \in [0,1]$ (i.e. $s(p,i)$ is a probability).
\end{enumerate}

It is important to note that $i$ need not be a real value and this
offers considerable modeling flexibility. In particular, the
interference dynamics need not be on the same time scale as the power
update dynamics. For a specific example, consider channel interference
that takes either ``high'' or ``low'' values as in the Gilbert channel
model \cite{Gilbert_1960}. In addition, assume that the channel
dynamics are twice as fast as the power update dynamics. Then in each
time slots, the interference state could take 4 values: (``low'',
``low''), (``low'', ``high''), (``high'', ``low''), (``high'',
``high''). Hence, this model allows for power updates that are not on
the same time scale as interference changes.

More generally, the interference fluctuates stochastically according
to a finite-state Markov chain \cite{Wang_FSMC_1995} with transition
matrix $P \in \Rbb^{I \times I}$. So if at time $t$ the interference
is at $i_t$, then $\Pbb(i_{t+1} = j | i_t = i) = P_{ij}$.  Without
loss of generality, we will enumerate the interference states as $\{1,
2, \hdots, I\}$. We assume that the interference fluctuations are
statistically independent of all other aspects of the model. This is a
very rich and descriptive model for stochastic interference but this
comes at a practical cost which will be discussed in Section
\ref{ssec:drawbacks}.

Although the interference states need not correspond to real-valued
interference levels, it is sometimes useful to take this point of
view. For example, we could take $s(p, i) = p/(p + i)$ or $s(p, i) = 1
- \exp(-p/i)$ which are qualitatively similar and have been used for
building dynamic models in the network power control literature
\cite{Kandukuri_PCMA_2000}. Emperical studies have shown that for
positive constants $\{\beta_i\}_{i=0}^2$, sigmoidal functions of the
form $s(p, i) = \left(1 - \frac{1}{2}\exp\left(-\beta_0 \frac{p}{i} +
    \beta_1\right)\right)^{\beta_2}$ provide a more accurate
description of actual wireless channels
\cite{Son_Sigmoid_2006}. Although different functional forms can be
used to model different types of channels,
Figure~\ref{fig:success_prob} shows that these functional forms can
actually be quite similar. The two properties given above are
sufficient for most of this paper but we will return to the issue of
choosing a particular $s(\cdot, \cdot)$ in Section~\ref{sec:SLBPC} and
Section~\ref{sec:performance}.

\begin{figure}[h]
  \centering
  \includegraphics[width=0.5\textwidth]{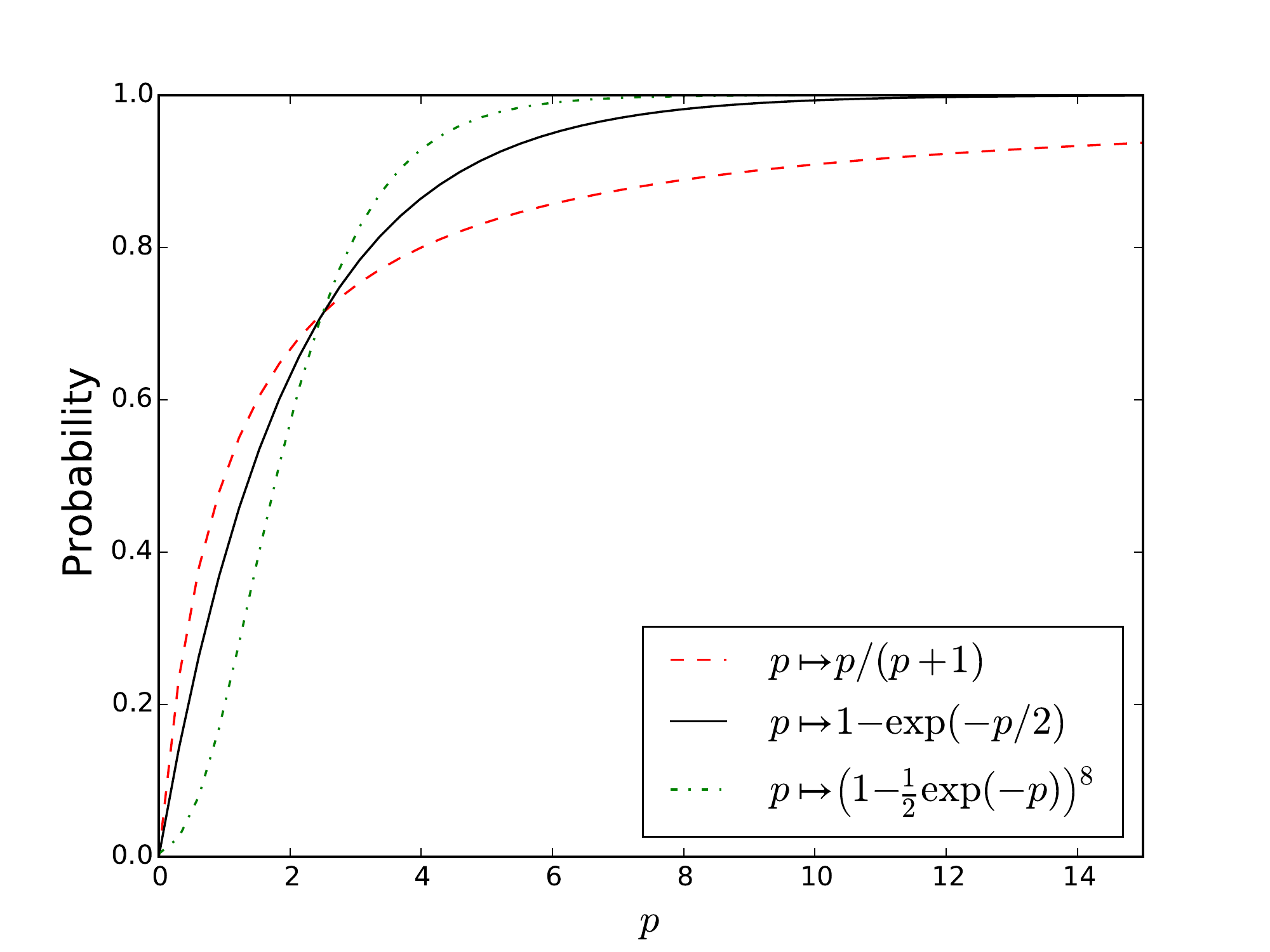}
  \caption{ Examples of different success probability
    functions\label{fig:success_prob}}
\end{figure}

\subsection{System Dynamics and Costs \label{ssec:dynamics}}
The transmitter is modeled as a packet buffer (job queue) with a
single transmitter (server). When a packet reaches the head-of-line
(HOL) it is assigned a residual deadline $d = D$. The transmitter
chooses a power $p$ from a finite set $\Pcal \subseteq \Rbb_{\geq 0}$
and transmits the packet. Assuming the interference state is $i$, with
probability $s(p, i)$ the transmission is successful. In this case,
the receiver will send an acknowledgment (ACK) to the transmitter over
a low rate control channel. The transmission is unsuccessful with
probability $1 - s(p, i)$. In this case, the receiver sends a negative
acknowledgment (NACK), the residual deadline $d$ is decremented by
one, and the transmitter reattempts the transmission. This continues
until either the packet is successfully transmitted or until the
residual deadline reaches zero. When the residual deadline reaches
zero, the HOL packet is ejected from the transmitter buffer. The next
packet becomes the new HOL packet and is assigned a residual deadline
of $D$, as before. Dropping the packet will temporarily degrade the
multimedia quality, but prevents jitter and delay. Note that the
receiver sends a receiver signal strength indication (RSSI) on the
control channel. As a result, we assume that the transmitter knows the
current interference state before a transmission power is chosen. This
assumption is relaxed in the performance evaluation in
Section~\ref{sec:performance}.

As the system evolves, the transmitter incurs the following non-negative costs in
each time slot. First, transmitting with power $p \in \Pcal$ costs $\Cp(p)$. We
assume that $\Cp(\cdot)$ is non-decreasing which reflects a natural desire to use
less power. Second, the transmitter incurs a backlog cost of $\Cb(b)$ when the
packet backlog is $b$. We assume that $\Cb(\cdot)$ is non-decreasing. This
creates what is known in the queueing literature as ``backlog pressure''
\cite{Walrand_Queueing_1988}. The backlog pressure reduces the overall latency
experienced by the sequence of packets. Finally, when a HOL packet is dropped
because the residual deadline reaches zero, the transmitter incurs a constant
cost of $\Cd$. This cost reflects an aversion to degrading the multimedia
stream.

In general, we will use subscripts to indicate time slots. For the packet
backlog, assume $b_0 = B$ and that the backlog at time $t$ is $b_t \in \{0, 1,
2, \hdots, B\}$ for all $t \geq 1$. The residual deadline at time $t$ is $d_t
\in \{1, 2, \hdots, D\}$. The power level at time $t$ (for now chosen according
to some arbitrary rule) is $p_t \in \Pcal$ and the interference level at time
$t$ is $i_t$. Let $S_t$ be a binary random variable representing the channel so
that $\Pbb(S_t = 0) = 1 - s(p_t, i_t)$ and $\Pbb(S_t = 1) = s(p_t, i_t)$. Using
$S_t$, we can write the cost in each time slot as follows:
\begin{equation}
  g(b_t, d_t, i_t, p_t)
  = \Ind{b_t > 0}\left(
    \Cp(p_t) + \Cb(b_t) + \Ind{S_t = 0}\Ind{d_t = 1}\Cd
  \right)
\end{equation}
The dynamics and costs are summarized in Figure~\ref{tab:model}. Recall that
the interference fluctuates independently of all other aspects of the model and
therefore $i_{t+1}$ is not influenced by $S_t$, $b_t$, or $d_t$.
\begin{table}
  \begin{center}
    \begin{tabular}{|l|c|c|}
      \hline
      & Stage Cost & State Transition\\
      \hline
      \hline
      $S_t = 1$, $b_t > 0$ & $\Cp(p_t) + \Cb(b_t)$ 
      & $(b_{t + 1}, d_{t+1}, i_{t+1}) = (b_t - 1, D, i_{t+1})$\\
      \hline
      $S_t = 0$, $d_t > 1, b_t > 0$ & $\Cp(p_t) + \Cb(b_t)$ 
      & $(b_{t + 1}, d_{t+1}, i_{t+1}) = (b_t , d_t - 1, i_{t+1})$\\
      \hline
      $S_t = 0$, $d_t = 1, b_t > 0,$ & $\Cp(p_t) + \Cb(b_t) + \Cd$
      & $(b_{t + 1}, d_{t+1}, i_{t+1}) = (b_t - 1, D, i_{t+1})$\\
      \hline
      $b_t = 0$ & 0 & Terminal State\\
      \hline
    \end{tabular}
  \end{center}
  \caption{ A table summarizing the costs and dynamics of the
    model. $S_t$ is a binary variable indicating whether the
    transmission in time slot $t$ is successful (1 if successful, 0
    otherwise). $\Cp(\cdot)$ and $\Cb(\cdot)$ are non-decreasing costs
    on the power and backlog respectively. $\Cd$ is a constant cost
    incurred when a deadline is violated. Note that when $b_t = 0$,
    the transmission is complete.\label{tab:model}}
\end{table}

There are some modeling aspects that we have chosen to exclude. For
example, we have not included a packet arrival process because we are
primarily interested in the transmission problem and hence focus on
the media packets that are already present at the transmitter. This
also reflects the initial motivation of a device transmitting a cached
video to another nearby user -- the video is already completely stored
and there are no arrivals in this case.\footnote{A Markovian arrival
  process could easily be incorporated into our formulation, but this
  would substantially complicate the mathematics in
  Section~\ref{ssec:proofs} without adding significant insight. In the
  case of arrivals, our control schemes could be applied by
  instantaneously sampling the current buffer level and responding
  based on the analysis and results presented here. We take this
  approach in the simulations presented in
  Section~\ref{sec:performance}.}

Another potential modeling aspect would be non-uniform deadlines. In
previous work \cite{Master_HOL_2014}, we allowed for the deadlines to
be a function of the backlog so that $D = D(b)$. The following
mathematical results still hold, but it is not clear that there is any
practical benefit to having backlog adaptive deadlines. For simplicity
of exposition, we assume that every packet has the same HOL deadline.

An alternative model could include deadlines for every job rather than
just for the HOL packet. Tracking all packet deadlines would cause the
state space to explode and the ``curse of dimensionality''
\cite{PutermanMDP} would yield such a problem numerically
intractable. This issue is discussed more in Section
\ref{ssec:drawbacks}. Furthermore, the HOL deadline allows us to
control inter-packet jitter and leads to more regularity in the
received packet stream. For applications like multimedia streaming,
this will lead to fewer video freezes and a more consistent user
experience at the receiver.

We should note that wireless video streaming involves many technical
issues such as error correction coding across multiple packets. For
example, packet level forward error correction was experimentally
investigated by Alay et al. \cite{Alay_FEC_2009}. They designed an
application layer mechanism to strategically insert parity packets,
thus allowing up to 50\% packet loss of video streams over wireless
channels with virtually no loss of quality. Hence, we can assume that
the receiver can recover the information in the packets that are
dropped due to the HOL deadlines. With this experimental work in mind,
our model is built to spotlight the theoretical trade-off between
transmitter power, packet drop rate (due to missed deadlines), and
backlog, which is key in these systems.

\subsection{Optimal Control Formulation\label{ssec:control}}
Given the Markovian nature of our model, we can now pose the power control
problem as a Markov Decision Process (MDP) \cite{PutermanMDP}. At time $t \in
\{0, 1, 2, \hdots\}$, our state is given by $x_t = (b_t, d_t, i_t)$. In each
time slot, a control policy uses the state to choose a power level for
transmission. Mathematically, our set of admissible policies is
\begin{equation}
  \Pi = 
  \{\pi : \{1, \hdots, B\} \times \{1, \hdots, D\} \times \{1, \hdots, I\} 
  \rightarrow \Pcal\}.
\end{equation}
We want to minimize the total expected cost. Since there is no arrival process,
the state $(0, d, i)$ is a terminal state with zero cost for any $(d, i)$. It
can take no more than $BD$ time slots to reach the class of terminal
states. Therefore, the optimal expected cost-to-go $\Jcal(b, d, i)$ is
well-defined:
\begin{equation}
  \Jcal(b, d, i) 
  = \min_{\pi \in \Pi} 
  \Ebb\left[
    \left. \sum_{t = 0}^\infty g(b_t, d_t, i_t, \pi(b_t, d_t, i_t)) \right| 
    (b_0, d_0, i_0) = (b, d, i)
  \right]
  \label{eq:Jcal_def}
\end{equation}
Note that in (\ref{eq:Jcal_def}), the minimization is implicitly
subject to the dynamics from Sections~\ref{ssec:channel} and
\ref{ssec:dynamics}. The dynamics along with the costs give us the
following Bellman equation:
\begin{align}
  \begin{split}
    \Jcal(b, d, i) 
    &= \min_{p \in \Pcal}
    \Bigg\{
      \Cb(b) + \Cp(p) \\
      &\hspace{-10mm}+\sum_{j} P_{ij}
      \Big[s(p, i)\Jcal(b-1, D, j)
      + (1 - s(p, i))(\Ind{d=1}(\Cd + \Jcal(b-1, D, j))
          + \Ind{d > 1}\Jcal(b, d-1, j))\Big]\Bigg\}            
  \end{split}
  \label{eq:bellman}
\end{align}
For each $(b, d, i)$, the minimizer is not necessarily unique so there may be
multiple optimal policies. We refer to \emph{the} optimal policy $\mu$ as the
policy defined by the following minimization:
\begin{align}
  \begin{split}
    \mu(b, d, i) &= \min \argmin_{p \in \Pcal}
    \Bigg\{
      \Cb(b) + \Cp(p) \\
      &\hspace{-10mm}+\sum_{j} P_{ij}
      \Big[s(p, i)\Jcal(b-1, D, j)
      + (1 - s(p, i))(\Ind{d=1}(\Cd + \Jcal(b-1, D, j))
          + \Ind{d > 1}\Jcal(b, d-1, j))\Big]\Bigg\}            
  \end{split}
  \label{eq:bellman_mu}
\end{align}
The $\argmin$ operator returns the set of minimizers and $\mu$ is defined by
taking the smallest minimizer. Therefore, $\mu$ is the optimal policy which uses
the smallest amount of power.

\subsection{Practical Drawbacks of Optimal Control\label{ssec:drawbacks}}
In principle, given the model, we can solve (\ref{eq:bellman}) and
(\ref{eq:bellman_mu}) with standard algorithms like value iteration
and policy iteration \cite{PutermanMDP} offline and then store $\mu$
as a look-up table for use online. However, there are some practical
limitations to this approach. The algorithmic complexity of computing
$\mu$ grows exponentially with the dimension of the state-space so it
can become difficult to compute the optimal power control for large
numbers of packets and large numbers of interference states. In
dynamic programming, this is known as the ``curse of dimensionality''
\cite{PutermanMDP}. If the system parameters change,
(\ref{eq:bellman}) and (\ref{eq:bellman_mu}) would have to be solved
again online, which would not be tractable. In addition, naively
storing $\mu$ would require $O(BDI)$ memory which would be prohibitive
for large packet buffers.

There are also problems with estimating a model in the first
place. Although finite-state Markov interference patterns are a useful
model \cite{Wang_FSMC_1995}, it can be difficult to estimate the
transition matrix $P$. First, there is the question of choosing
quantization levels. Because of the algorithmic complexity mentioned
above, there is a trade-off between model fidelity and ease of
computation. In addition, as we increase $I$, the ``stochastic
complexity'' \cite{Barron_MDL_1998} of estimating the model
increases. In brief, the number of parameters to be estimated is
$O(I^2)$ and the number of observations needed for reliable estimation
grows exponentially with the number of parameters. Finally, note that
the finite-state Markov channel is a good stochastic model for
interference fluctuating over time, but it does not necessarily
capture the notion of fluctuations over space. In particular,
interference patterns can be very different in different places which
would require estimating new models as mobile users move to new
locations.

\section{An Illustrative Case \label{sec:illustrative}}
Because of the practical limitations of the general optimal control
approach, we instead consider a special case with fixed
interference. This corresponds to taking $I = 1$. The goal is to
mathematically understand the structure of the special case and use
this understanding to overcome the aforementioned complexity issues in
the general case.  We continue to use $\Jcal$ to represent the
expected cost-to-go and $\mu$ to represent the optimal policy, but we
no longer include the third argument because the interference is
constant. We continue to keep interference as the second argument to
$s(\cdot, \cdot)$. The new Bellman equation is given below:
\begin{align}
  \begin{split}
    \Jcal(b, d) = \min_{p \in \Pcal}
    \Big\{
      &\Cb(b) + \Cp(p)\\
      &+s(p, i)\Jcal(b-1, D)
      + (1 - s(p, i))(\Ind{d=1}(\Cd + \Jcal(b-1, D))
          + \Ind{d > 1}\Jcal(b, d-1))\Big\}            
  \end{split}
  \label{eq:simple_bellman}
\end{align}

First we numerically explore the optimal power control of this special case. We
demonstrate and explain the salient features of the optimal policy. Then we move
on to prove theoretical properties of the optimal policy. The theoretical
results in this section will be used to design a power control scheme which
approximates the optimal control.

\subsection{Numerical Results\label{ssec:numerical}}
In this section we fix the following parameters while varying $\Cd$: $\Pcal =
\{2, 4, 6\}$, $s(p, i) = 1 - \exp(-p/2)$, $\Cb(b) = b$, $\Cp(p) = p$, $B = 20$,
$D = 5$. Varying $\Cd$ corresponds to varying the aversion to dropping a packet
relative to the other costs. A low value of $\Cd$ corresponds to accepting a low
quality stream while a high value of $\Cd$ corresponds to requiring a high
fidelity stream. These parameters do not reflect a particular physical system
but are intended to illustrate the features of $\mu$.

In Figure~\ref{fig:policies}, we plot $\mu$ for $\Cd \in \{1, 10, 100\}$. For
each $(b, d) \in \{1, 2, \hdots, B\} \times \{1, 2, \hdots, D\}$, we plot a
point indicating $\mu(b, d) \in \Pcal$. We first note that for all values of
$\Cd$ and each $d$, $b \mapsto \mu(b, d)$ is non-decreasing. This is expected
because the backlog pressure incentivizes the transmitter to ``try harder'' to
transmit the packets.

Varying $\Cd$ does cause qualitatively different behavior of $d \mapsto \mu(b,
d)$. When $\Cd = 1$, $d \mapsto \mu(b, d)$ is non-decreasing for every $b$. This
shows that when the cost of dropping a packet is low, the transmitter will
``give up'' as the deadline approaches. In contrast, when $\Cd = 100$, $d
\mapsto \mu(b, d)$ is non-increasing for every $b$. In this case, the cost of
dropping a packet is high so the transmitter ``tries harder'' as the deadline
approaches. When $\Cd = 10$, we see that the ``give up'' and ``try harder''
behavior depends on the value of $b$. Indeed for $\Cd = 10$, the policy adopts a
``try harder'' behavior for low values of $b$ and a ``give up'' behavior for
large values of $b$. In the following section, we analytically determine the
relationship between the costs that causes this qualitatively different
behavior.

\begin{figure}
  \centering
  \begin{subfigure}{0.5\textwidth}
    \includegraphics[width=\textwidth]{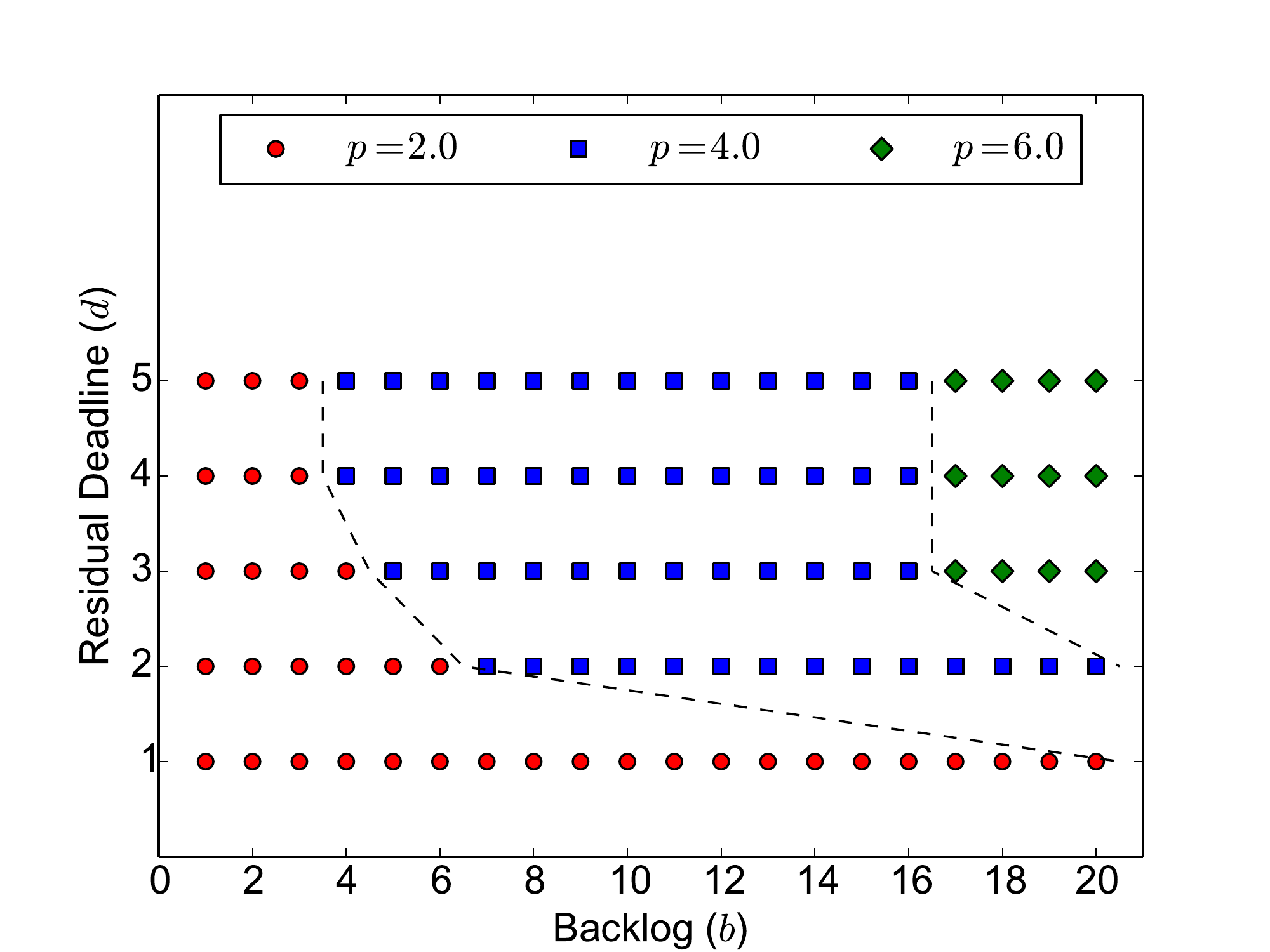}
    \caption{ $\Cd = 1$}
  \end{subfigure}%
  \begin{subfigure}{0.5\textwidth}
    \includegraphics[width=\textwidth]{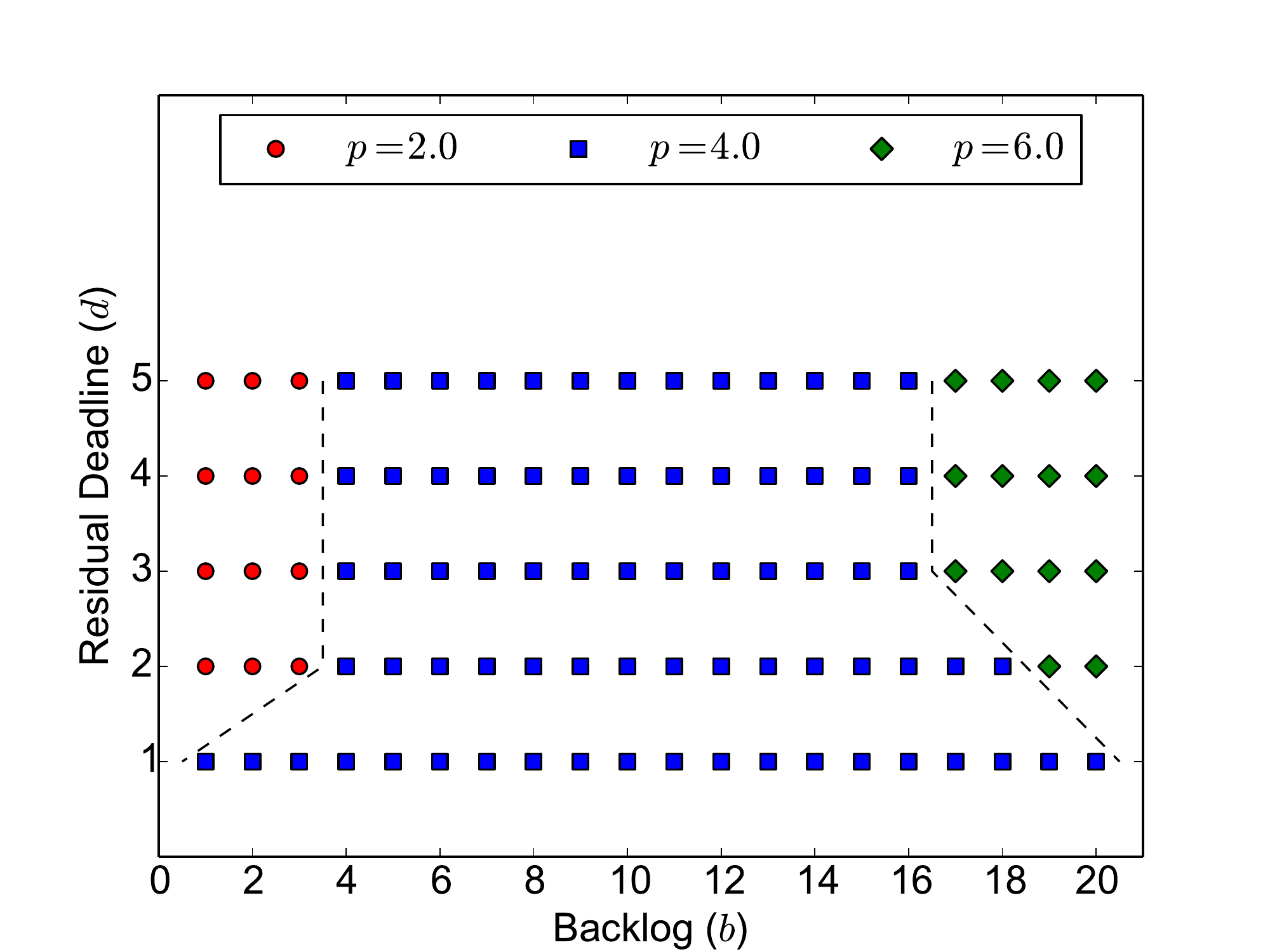}
    \caption{ $\Cd = 10$}
  \end{subfigure}\\
  \begin{subfigure}{0.5\textwidth}
    \includegraphics[width=\textwidth]{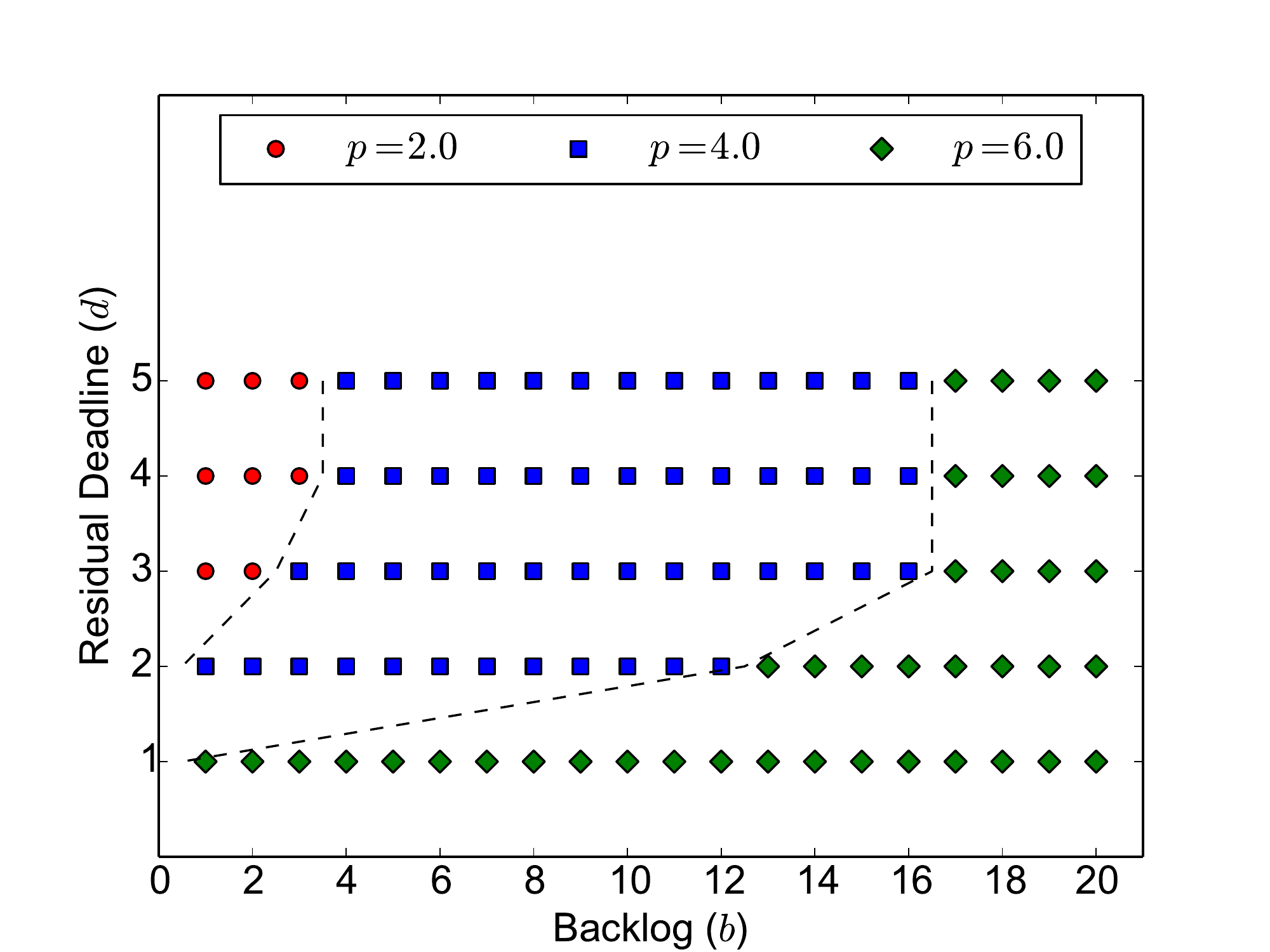}
    \caption{ $\Cd = 100$}
  \end{subfigure}
  \caption{ Example policies for the illustrative case presented in
    Section~\ref{sec:illustrative}. We take $\Pcal = \{2, 4, 6\}$,
    $s(p, i) = 1 - \exp(-p/2)$, $\Cb(b) = b$, $\Cp(p) = p$, $B = 20$,
    $D = 5$ and vary $\Cd$. For each $(b, d) \in \{1, 2, \hdots, B\}
    \times \{1, 2, \hdots, D\}$, we plot a point indicating $\mu(b, d)
    \in \Pcal$. The dashed lines indicate partitions in the state
    space that separate the different power values. For any $\Cd$ and
    any $d$, we see that $b \mapsto \mu(b, d)$ is
    non-decreasing. However, varying $\Cd$ does affect how $\mu(b, d)$
    varies with $d$. When $\Cd$ is low ($\Cd = 1$), $d \mapsto \mu(b,
    d)$ is non-decreasing for every $b$. When $\Cd$ is high ($\Cd =
    100$), $d \mapsto \mu(b, d)$ is non-increasing for every $b$. When
    $\Cd$ is moderate ($\Cd = 10$), $d \mapsto \mu(b, d)$ is
    non-decreasing for certain values of $b$ and non-increasing for
    others.    \label{fig:policies}}
\end{figure}

\subsection{Theoretical Results \label{ssec:proofs}}
The main results can be summarized informally as follows:
\begin{itemize}
\item For any $s(\cdot, \cdot)$, $b \mapsto \mu(b, d)$ is
  non-decreasing. 
\item For any $s(\cdot, \cdot)$, if $\Cb(b) + \min_{p \in
    \Pcal}\{\Cp(p) - s(p, i)\Cd\} \geq 0$ then $d \mapsto \mu(b, d)$
  is non-decreasing.
\item For any $s(\cdot, \cdot)$, if $\Cb(b) + \min_{p \in
    \Pcal}\{\Cp(p) - s(p, i)\Cd\} \leq 0$ then $d \mapsto \mu(b, d)$
  is non-increasing.
\item When $p \mapsto s(p, i)$ is concave or sigmoidal and
  $\Cp(\cdot)$ is linear, $b \mapsto \mu(b, d)$ is well-approximated
  by a sublinear function of $\Cb(b)$. Roughly speaking, the optimal
  transmitter power grows sublinearly with the backlog cost.
\end{itemize}

To rigorously prove these theoretical results, we need to apply the
following theorem (a version of what is known as Topkis's Theorem
\cite{Topkis_Minimizing_1978}) which is standard in the stochastic
dynamic programming literature, e.g. \cite{PutermanMDP}.
\begin{theorem}
  Take some finite $\Ycal \subseteq \Rbb$. We say $f:\Rbb \times
  \Ycal \rightarrow \Rbb$ is submodular if for any $x^- \leq x^+$
  and $y^- \leq y^+$, the following inequality holds:
  \begin{equation}
    f(x^+, y^+) + f(x^-, y^-) \leq f(x^+, y^-) + f(x^-, y^+)
    \label{eq:submodular}
  \end{equation}
  Let $g(x) = \min\argmin_{y \in \Ycal} f(x, y)$. If $f$ is
  submodular then $g$ is non-decreasing.
  \label{thrm:topkis}
\end{theorem}
Theorem~\ref{thrm:topkis} and its generalizations have been used to
prove monotonicity properties of MDPs in the queueing and control
literature (e.g. \cite{Serfozo_Monotone_1976},
\cite{Veatch_Monotone_1992}, and \cite{Dragut_Lattice_2010}). However,
the HOL deadlines in our model prevent us from directly applying these
previous results. As seen in Section~\ref{ssec:numerical}, the
monotonicity of the optimal control policy varies with the backlog. As
a result, we need to approach this problem from first principles. We
begin with some definitions and notation which will allow us to
reformulate the Bellman equation in (\ref{eq:simple_bellman}).
\begin{definition}
  We define the following auxiliary quantities. For each $b \geq 1$, 
  $$\delta(b,1) = \Cb(b) + \min_{p \in \Pcal}\left\{\Cp(p)- s(p, i)\Cd\right\}.$$
  For $d > 1$, $\delta(b, d)$ is inductively defined as follows:
  \begin{align}
  \delta(b,d) 
  &= \Cb(b) + 
  \min_{p \in \Pcal}\left\{\Cp(p) 
    - s(p, i)\left(\Cd + \sum_{k=1}^{d-1}\delta(b,k)\right)\right\}
  \end{align}
  In addition, we define $ \sigma(b,d) = \sum_{k=1}^d \delta(b,k)$.  For each of
  these definitions, whenever the upper limit on a sum is less than the lower
  limit, the sum is equal to zero.

  For $b \geq 1$ we define the operator $T_b: \Rbb \rightarrow \Rbb$ as
  follows:
  $$T_b(x) = x + \Cb(b) + \min_{p \in \Pcal}\left\{\Cp(p) - s(p, i)(\Cd + x)\right\}$$
\end{definition}

Note that $T_b(0) = \Cb(b) + \min_{p \in \Pcal}\left\{\Cp(p) - s(p,
  i)\Cd\right\}$. Therefore, we will be using the sign of $T_b(0)$ to tell us
whether $\mu(b, \cdot)$ is increasing or decreasing. This idea is further
elucidated in the following propositions.

\begin{proposition}
  \label{prop:J_mu}
  For $b \geq 1$, the Bellman equation in (\ref{eq:simple_bellman})
  can be reformulated as follows:
  \begin{align}
    \begin{split}
      \Jcal(b,d)
      &= \left\{
        \begin{array}{ll}
          \Cd + \Jcal(b-1, D) + \delta(b,1) &, d = 1\\
          \Jcal(b, d-1) + \delta(b,d) &, d > 1
        \end{array}
      \right.
    \end{split}
  \end{align}
  Consequently, the optimal policy can we written as
  \begin{align}
    \mu(b,d) &= \min \argmin_{p \in \Pcal}\left\{\Cp(p) - s(p, i)(\Cd
      + \sigma(b,d-1))\right\}.
  \end{align}
  Also note that since every cost is non-negative, $\Jcal(b,d) \geq 0$
  for every $(b,d)$.
\end{proposition}
\begin{proof}
  We apply the principle of strong mathematical induction.  For $d =
  1$, this is just a reordering of the terms in the Bellman equation.
  For $d = 2$, we have the following:
  \begin{align}
    \Jcal(b,2)
    &= \min_{p \in \Pcal}\{\Cp(p) + \Cb(b) 
    + s(p, i)\Jcal(b-1, D) + (1-s(p, i))\Jcal(b,1)\}\\
    &= \Jcal(b,1) + \Cb(b)
    + \min_{p \in \Pcal}\{\Cp(p) - s(p, i)(\Jcal(b,1) - \Jcal(b-1, D))\}\\
    &= \Jcal(b,1) + \Cb(b) + 
    \min_{p \in \Pcal}\{\Cp(p) - s(p, i)(\Cd + \delta(b,1))\}\\
    &= \Jcal(b,1) + \delta(b,2)
  \end{align}
  Now assume the proposition is true for all $d - 1, d-2, d-3, \hdots, 2, 1$. 
  \begin{align}
    &\Jcal(b,d)\nonumber\\
    &= \min_{p \in \Pcal}\{\Cp(p) + \Cb(b) 
    + s(p, i)\Jcal(b-1, D) + (1-s(p, i))\Jcal(b,d-1)\}\\
    &= \Jcal(b,d-1) + \Cb(b)
    +\min_{p \in \Pcal}\{\Cp(p) - s(p, i)(\Jcal(b,d-1) - \Jcal(b-1, D))\}\\
    &= \Jcal(b,d-1) + \Cb(b) 
    +\min_{p \in \Pcal}\left\{\Cp(p)
      - s(p, i)\left(\sum_{k=2}^{d-1}(\Jcal(b, k) - \Jcal(b, k-1))
        + \Jcal(b,1) - \Jcal(b-1, D)\right)\right\}\label{eq:telescoping}\\
    &= \Jcal(b,d-1) + \Cb(b) 
    +\min_{p \in \Pcal}\left\{\Cp(p)
      - s(p, i)\left(\sum_{k=2}^{d-1}\delta(b,k) 
        + \Jcal(b,1) - \Jcal(b-1, D)\right)\right\}\label{eq:induction1}\\
    &= \Jcal(b,d-1) + \Cb(b) + 
    \min_{p \in \Pcal}\left\{\Cp(p)
      - s(p, i)\left( \Cd + \sum_{k=1}^{d-1}\delta(b,k) \right)\right\}\label{eq:induction2}\\
    &= \Jcal(b,d-1) + \delta(b,d)
  \end{align}
  In (\ref{eq:telescoping}), we replace $\Jcal(b, d-1)$ with an
  equivalent telescoping sum. This allows us to apply the induction
  hypothesis in (\ref{eq:induction1}) and (\ref{eq:induction2}).  

  This reformulation of the Bellman equation allows us to just ignore
  the terms that don't involve $p$ to see that
  $$
  \mu(b,d) = \min\argmin_{p \in \Pcal} \left\{ \Cp(p) - s(p, i)\left(\Cd +
      \sum_{k=1}^{d-1}\delta(b,k)\right) \right\} .$$ 
  Replacing the sum with $\sigma(b,d-1)$ completes the proof.
\end{proof}

This gives us a semi-analytic expression for $\mu$. In other words, we now have
an expression for $\mu$ which depends on $\sigma$ rather than on $\Jcal$. We will
use the $\delta$ terms to determine whether $\sigma(b, d)$ is increasing or
decreasing in $b$ and/or $d$.
\begin{proposition}
  \label{prop:Tb}
  For each $b$, $T_b(\cdot)$ is non-decreasing. We can write $\sigma(b,d) =
  T_b(\sigma(b,d-1))$ where $\sigma(b,0) = 0$.
  For $d \geq 1$ and $b \geq 1$ we have the following conditions:
  \begin{itemize}
  \item For each $d$, $\sigma(\cdot, d)$ is non-decreasing.
  \item If $T_b(0) \geq 0$ then $\sigma(b,\cdot)$ is non-negative and
    non-decreasing.
  \item If $T_b(0) \leq 0$ then $\sigma(b,\cdot)$ is non-positive and
    non-increasing.
  \end{itemize}
\end{proposition}
\begin{proof}
  Take $x \geq y$. For any $p$, $(1 - s(p, i)) \geq 0$ so $(1 - s(p, i))x \geq (1 -
  s(p, i))y$. Adding terms to both sides preserves the inequality so
  \begin{align}
    \Cb(b) + (1-s(p, i))x + \Cp(p) - s(p, i)\Cd \geq \Cb(b) + (1-s(p, i))y + \Cp(p) -
    s(p, i)\Cd.
  \end{align} 
  Taking the minimum over $p \in \Pcal$ and using the monotonicity of minimization
  gives us that $T_b(x) \geq T_b(y)$.

  We originally defined $\sigma(b,d) = \sum_{k=1}^d \delta(b,k)$. If $d = 0$
  then the sum is vacuous so $\sigma(b,0) = 0$. For $d \geq 1$,
  \begin{align} 
    \sigma(b,d) 
    &= \sum_{k=1}^d \delta(b,k) 
    = \sum_{k=1}^{d-1} \delta(b,k) + \delta(b,d)
    = \sigma(b,d-1) + \delta(b,d)\\
    &= \sigma(b,d-1) + \Cb(b) 
    + \min_{p \in \Pcal}\left\{\Cp(p) - s(p, i)\left(\Cd +
        \sum_{k=1}^{d-1}\delta(b,k)\right)\right\}\\ 
    &= \sigma(b,d-1) + \Cb(b) 
    + \min_{p \in \Pcal}\left\{\Cp(p) - s(p, i)\left(\Cd +
        \sigma(b,d-1)\right)\right\}\\ 
    &= T_b(\sigma(b,d-1))
  \end{align}

  Since $\Cb(\cdot)$ is non-decreasing, $T_{b'}(x) \geq T_b(x)$ for
  any $b' \geq b$. Since $\sigma(b, d)$ can be found by iteratively
  applying $T_b(\cdot)$ to $0$, we have that $\sigma(b', d) \geq
  \sigma(b, d)$ for any $b' \geq b$. Therefore, $\sigma(\cdot, d)$ is
  non-decreasing for all $d$.

  Now we show how $\sigma(b, d)$ varies with $d$. Assume that $T_b(0)
  \geq 0$ and proceed by induction. $\sigma(b, 0) = 0$ so we have that
  $\sigma(b, 1) = T_b(\sigma(b, 0)) \geq \sigma(b, 0)$. If $\sigma(b,
  d) \geq \sigma(b, d-1)$ then applying $T_b$ and using the fact that
  $T_b$ is non-decreasing gives us that
  \begin{align}
    \sigma(b, d+1) = T_b(\sigma(b, d)) &\geq T_b(\sigma(b, d-1)) =\sigma(b, d)
  \end{align}
  Therefore, $\sigma(b, \cdot)$ is non-decreasing for all $d$. $\sigma(b,0) = 0$
  so $\sigma(b, d)$ is non-negative. The case of $T_b(0) \leq 0$ is analogous.
\end{proof}



Now that we understand the monotonicity properties of $\sigma(b, d)$, we can
leverage these to understand the monotonicity properties of $\mu(b, d)$.
\begin{proposition}
  \label{prop:mainResult}
  The optimal policy $\mu(b,d)$ is always non-decreasing in $b$. If
  $T_b(0) \geq 0$ then $\mu(b,d)$ is non-decreasing in $d$. If $T_b(0)
  \leq 0$ then $\mu(b,d)$ is non-increasing in $d$.
\end{proposition}
\begin{proof}

  We can now leverage Theorem~\ref{thrm:topkis} for our specific problem.  From
  Proposition~\ref{prop:J_mu},
  $$\mu(b,d) = \min\argmin_{p \in \Pcal}\left\{\Cp(p) - 
    s(p, i)(\Cd + \sigma(b,d-1))\right\}.$$ 

  Let $f(x, p) = \Cp(p) - s(p, i)x$ and define $g(x)$ as in
  Theorem~\ref{thrm:topkis} with $\Ycal = \Pcal$. Then $\mu(b, d) =
  g(\Cd + \sigma(b, d-1))$. Let LHS and RHS be the left-hand side and
  right-hand side of (\ref{eq:submodular}).
  \begin{align}
    LHS &- RHS\nonumber\\
    &= \left(\Cp(p^+) - s(p^+, i)x^+ + \Cp(p^-) - s(p^-, i)x^-\right)
    -\left(\Cp(p^+) - s(p^+, i)x^- + \Cp(p^-) - s(p^-, i)x^+\right)\\
    &= (s(p^+, i) - s(p^-, i))(x^- - x^+)
  \end{align}
  By assumption, $(s(p^+, i) - s(p^-, i)) \geq 0$ while $(x^- - x^+) \leq 0$ so
  $LHS \leq RHS$. Therefore, $f$ is submodular, $g$ is non-decreasing, and
  $\mu(b,d)$ is a non-decreasing function of $\sigma(b, d-1)$. 

  $\sigma(b, d)$ is non-decreasing in $b$ so $\mu(b,d)$ is non-decreasing in
  $b$. If $T_b(0) \geq 0$ then $\sigma(b, d)$ is non-decreasing in $d$ so
  $\mu(b,d)$ is non-decreasing in $d$. If $T_b(0) \leq 0$ then $\sigma(b,d)$ is
  non-increasing in $d$ so $\mu(b,d)$ is non-increasing in $d$.
\end{proof}

In order to understand more precisely how $\mu(b, d)$ increases with $b$, we
need to understand more precisely how $\sigma(b, d)$ increases with $b$. The
following result gives us that understanding.
\begin{proposition}
  \label{prop:OCb}
  Let $m = \min_{p \in \Pcal}\{ \Cp(p) - s(p, i)\Cd\}$ and $M = \max_{p \in
    \Pcal}\{ \Cp(p) - s(p, i)\Cd\}$. Let $s_{min} = s(\min\Pcal, i)$ and
  $s_{max} = s(\max\Pcal, i)$.  For any fixed $b$ and $d$, we have the following
  cases:
  \begin{itemize}
  \item $T_b(0) > 0$:
    \begin{equation}
      \left(\Cb(b) + m\right)\sum_{k=0}^{d-1} (1 - s_{max})^k
      \leq \sigma(b, d) \leq
      \left(\Cb(b) + M\right)\sum_{k=0}^{d-1} (1 - s_{min})^k
    \end{equation}
  \item $T_b(0) < 0$:
    \begin{equation}
      \left(\Cb(b) + m\right)\sum_{k=0}^{d-1} (1 - s_{min})^k
      \leq \sigma(b, d) \leq
      \left(\Cb(b) + M\right)\sum_{k=0}^{d-1} (1 - s_{max})^k
    \end{equation}
  \item $T_b(0) = 0$:
    \begin{equation}
      \sigma(b, d) = 0
    \end{equation}
\end{itemize}
\end{proposition}
\begin{proof}
  By substituting the optimal policy into the minimization in $T_b$ we have
  \begin{align}
    \sigma(b,d+1) &= T_b(\sigma(b,d))
    = \sigma(b,d) + \Cb(b) 
    + \min_{p \in \Pcal} \left\{\Cp(p) - s(p, i)(\Cd + \sigma(b,d))\right\}\\
    &= \sigma(b,d) + \Cb(b) + \Cp(\mu(b,d)) - s(\mu(b,d), i)(\Cd + \sigma(b,d))\\
    &= (1 - s(\mu(b,d), i))\sigma(b,d) + \Cb(b) + \Cp(\mu(b,d)) - s(\mu(b,d), i)\Cd
  \end{align}

  Consider the case when $T_b(0) > 0$ so that $\sigma(b, d) > 0$. By
  first optimizing $\Cp(p) - s(p, i)\Cd$ over $p \in \Pcal$ and then
  optimizing $(1 - s(p, i))$ over $p \in \Pcal$, we have the following
  inequality:
  \begin{equation}
    (1 - s_{max})\sigma(b,d) + \Cb(b) + m 
    \leq \sigma(b,d+1) \leq 
    (1 - s_{min})\sigma(b,d) + \Cb(b) + M
  \end{equation}
  For $d = 1$, the condition holds because $\Cb(b) + m = \sigma(b, 1) \leq \Cb(b)
  + M$. If the condition holds for $\sigma(b, d)$ then we have that
  \begin{align}
    \sigma(b, d + 1) 
    &\leq (1 - s_{min})\sigma(b, d) + \Cb(b) + M
    \leq (1 - s_{min})\left( (\Cb(b) + M)\sum_{k=0}^{d-1}(1 - s_{min})^k\right)
    + \Cb(b) + M\\
    &= (\Cb(b) + M)\sum_{k=0}^{d}(1 - s_{min})^k.
  \end{align}
  We can apply similar reasoning for the lower bound. By induction, the condition
  holds for all $d$.

  The case when $T_b(0) < 0$ and $\sigma(b, d) < 0$ is analogous. 

  Recall that $\sigma(b, d)$ can be computed by iteratively applying
  $T_b(\cdot)$ to 0. When $T_b(0) = 0$, we have that $\sigma(b, d) = 0$ for all
  $d$.
\end{proof}

Proposition~\ref{prop:OCb} tells us that $b \mapsto \sigma(b, d)$ is
bounded above and below by affine functions of $\Cb(b)$. Intuitively,
the important thing to note is that for any fixed $d$, $\sigma(b, d)$
is roughly proportional to $\Cb(b)$. For example, in the degenerate
case when $\Pcal$ is a singleton, the upper and lower bounds are equal
and $\sigma(b,d)$ is an affine function of $\Cb(b)$.  Furthermore, the
inequalities hold for any $b$ no matter how large $B$ is.

In the previous propositions, we only assumed that $\Cp(\cdot)$ and
$s(\cdot, i)$ were non-decreasing. Now we consider when $\Cp(\cdot)$
is linear and $s(\cdot, i)$ is concave. Assuming $\Cp(\cdot)$ is
linear is somewhat limiting, but linear power costs form an important
class that has been considered in the communications engineering
literature (e.g. \cite{Cruz_INFOCOM_2003, Goyal_INFOCOM_2003}) as well
as in the information theory literature
(e.g. \cite{Bettesh_IEEEIT_2006}). In addition, although assuming
$s(\cdot, i)$ is concave may seem limiting, we will use this as a step
towards considering a much broader class of success probability
functions. We also introduce the following definition which allows us
to precisely state the necessary result.
\begin{definition}
  We say a non-negative real-valued function $f$ is sublinear if 
  \begin{enumerate}[(i)]
  \item there exists $K_1, K_2 \in (0, \infty)$ such that $K_1x + K_2
    \geq f(x)$ for all $x \geq 0$, and
  \item the asymptotic behavior of $f$ is such that
    $$\lim_{x \rightarrow \infty} \frac{f(x)}{x} = 0.$$
  \end{enumerate}
\end{definition}
This definition differs slightly from the standard computer science definition
of sublinearity \cite{CLRS_2001} which would only require the second
condition. The second condition only considers the asymptotic behavior of the
function while the first condition also captures the behavior over finite
subsets of the domain. This is important in our application because we are
minimizing over a finite set of power values $\Pcal$.

\begin{proposition}
  \label{prop:approx}
  Consider the following function:
  $$
  \gamma(b, d)
  = \argmin_{p \geq 0}\{ \Cp(p) - s(p, i)(\Cd + \sigma(b, d-1))\}
  $$
  Assume $\Cp(p) = Kp$ for some $K > 0$ and that for some function
  $\phi:[0,\infty) \rightarrow [0,\infty)$ that is differentiable,
  strictly increasing, and convex function, that
  $\frac{\partial}{\partial p}s(p, i) = \frac{1}{\phi(p, i)}$. Then,
  $\gamma(b, d)$ is a sublinear function of $\sigma(b, d-1)$.
\end{proposition}
\begin{proof}
  First assume that $\Cd + \sigma(b, d-1) \geq 0$. This is always the
  case when $T_b(0) \geq 0$. The function being minimized is convex so
  the first-order optimality condition is both necessary and
  sufficient.
  $$
  \frac{\partial}{\partial p}\{ \Cp(p) - s(p, i)(\Cd + \sigma(b, d-1))\} 
  = K - s'(p, i)(\Cd + \sigma(b, d-1)) = 0
  \iff s'(p, i) = \frac{K}{\Cd + \sigma(b, d-1)}
  $$
  $\phi(\cdot, i)$ is strictly increasing so it is invertible. Therefore, 
  $$
  \gamma(b, d) = \phi^{-1}\left(\frac{\Cd + \sigma(b, d-1)}{K}, i\right)
  $$
  where we take the inverse with respect to the first argument.
  $\phi(\cdot, i)$ is strictly increasing and convex so
  $\phi^{-1}(\cdot, i)$ is strictly increasing and concave.
  Therefore, $\phi^{-1}(\cdot, i)$ is sublinear.

  When $\Cd + \sigma(b, d) < 0$, we are minimizing an increasing function so
  $\gamma(b, d) = 0$ which is trivially sublinear.
\end{proof}

We now briefly consider how one could adapt
Proposition~\ref{prop:approx} for the case that $s(\cdot, i)$ is
sigmoidal. Recall that while concave forms for $s(\cdot, i)$ have been
used in the transmitter power control literature
(e.g. \cite{Kandukuri_PCMA_2000}), sigmoidal forms for $s(\cdot, i)$
have been emperically shown to be more accurate
(e.g. \cite{Son_Sigmoid_2006}). A sigmoidal function of the form $p
\mapsto \left(1 - \frac{1}{2}\exp\left(-\beta_0 \frac{p}{i} +
    \beta_1\right)\right)^{\beta_2}$ will be strictly convex for
smaller values of $p$ and strictly concave for larger values of
$p$. As a result, two distinct values of $p$ will satisfy the
first-order optimality criterion used in the proof of
Proposition~\ref{prop:approx}.

To alleviate this problem, we can adjust the definition of $\gamma(b,
d)$. Instead of using
$$
\gamma(b, d)
= \argmin_{p \geq 0}\{ \Cp(p) - s(p, i)(\Cd + \sigma(b, d-1))\}
$$
we can use
$$
\gamma(b, d)
= \argmin_{p \geq 0}\{ \Cp(p) - s_{ccv}(p, i)(\Cd + \sigma(b, d-1))\}
$$
where $s_{ccv}(\cdot, i)$ is the \emph{concave envelope} of $s(\cdot, i)$,
i.e.
$$s_{ccv}(\cdot, i) = \inf\{h(\cdot) : s(p, i) \leq h(p)\,\forall\, p \geq 0,\, h(\cdot) \hbox{ is concave}\}.$$
The concave envelope of a function is analogous to the convex envelope
of a function and such envelopes be used to approximately minimize
non-convex (or maximize non-concave) function (see
e.g. \cite{Boyd_CVXBook}). However, in this case, using
$s_{ccv}(\cdot, i)$ instead of $s(\cdot, i)$ yields an exact result
rather than an approximation. Indeed, we are using $s_{ccv}(\cdot, i)$
instead of $s(\cdot, i)$ only in the definition of $\gamma(b,d)$ and
not in the definitions of $\delta(b, d)$ or $\sigma(b, d)$. As a
result, $s_{ccv}(\cdot, i)$ does not change the minimum but merely
makes the first-order optimization condition sufficient for finding
this minimum. This is depicted in Figure~\ref{fig:envelope}. 

Furthermore, for a given sigmoidal function $s(\cdot, i)$, it is easy
to compute $s_{ccv}(\cdot, i)$. By definition, $s_{ccv}(\cdot, i)$ is
the smallest concave function that majorizes $s(\cdot, i)$. Since
$s(\cdot, i)$ is concave for sufficiently large values, there is some
$p^*$ such that $s_{ccv}(p, i) = s(p, i)$ for all $p \geq p^*$. In
addition, because $s(\cdot, i)$ is convex for $p \in [0, p^*]$, the
concave envelope will be linear on $[0, p^*]$. Hence,
$$s_{ccv}(p, i)
= \left\{
  \begin{array}{cl}
    K_{ccv}p + s(0, i) &, p \in [0, p^*)\\
    s(p, i) &, p \geq p^*
  \end{array}
\right.
$$
where $p^* \in (0, \infty)$ satisfies
$$s'(p^*, i) = \frac{s(p^*, i) - s(0, i)}{p^*}$$
and $K_{ccv} = s'(p^*, i)$. Figure~\ref{fig:envelope} shows this
graphically.

This demonstrates that (with a slight modification)
Proposition~\ref{prop:approx} applies to both concave as well as
sigmoidal success probability functions. Proposition~\ref{prop:OCb}
tells us that $\sigma(b, d)$ is roughly proportional to $\Cb(b)$ so
Proposition~\ref{prop:approx} indicates that $\gamma$ is
well-approximated by a sublinear function of $\Cb(b)$.  Introducing
$\phi$ and $\gamma$ allows us to to make use of continuous
optimization theory which leads to closed-form results that wouldn't
arise when optimizing over the finite set $\Pcal$. We will see below
that $\gamma$ is useful for designing power control schemes.

\begin{figure}
  \centering
  \begin{subfigure}{0.5\textwidth}
    \includegraphics[width=\textwidth]{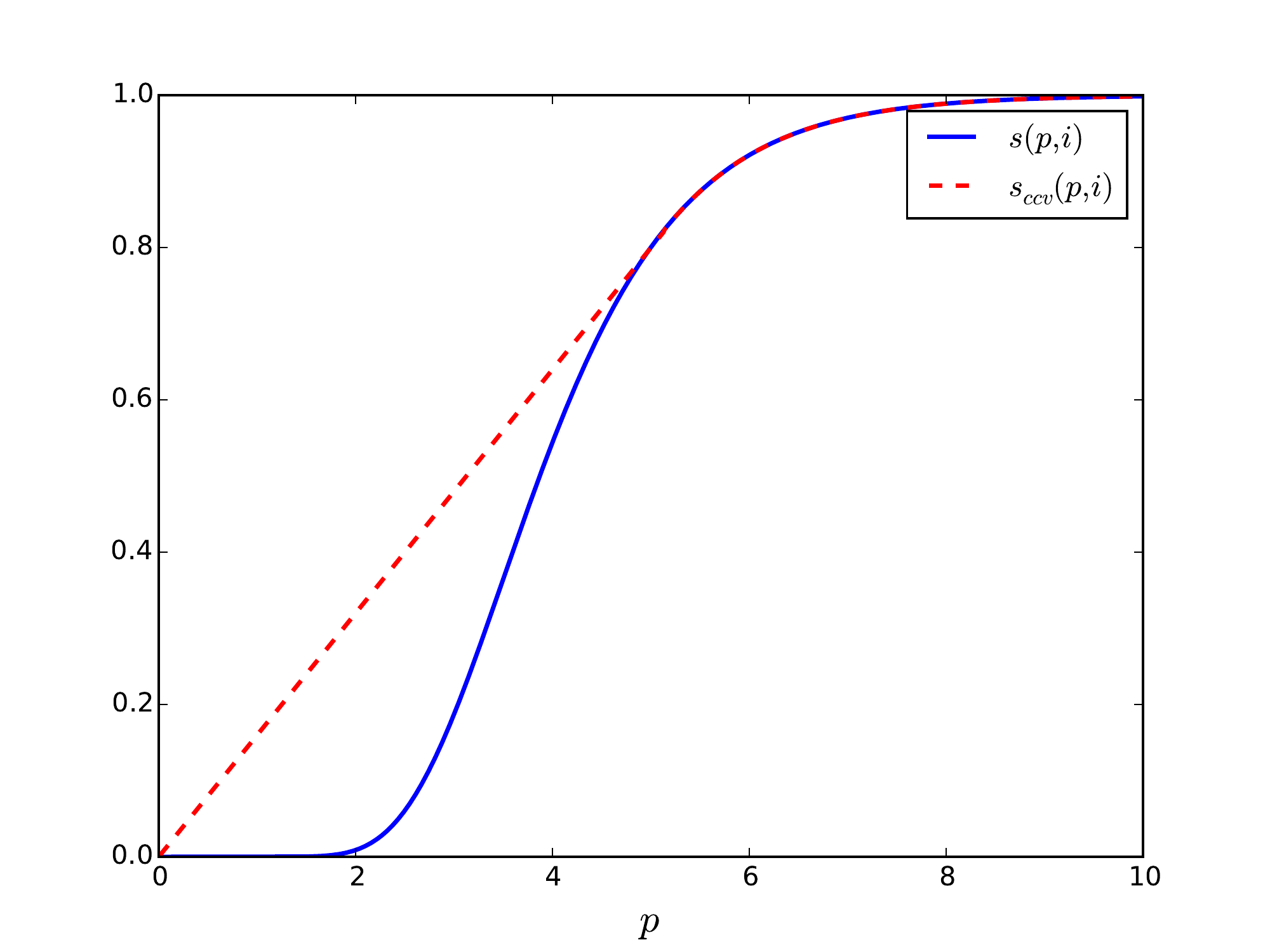}
    \caption{
      \label{fig:ccv-env}}
  \end{subfigure}%
  \begin{subfigure}{0.5\textwidth}
    \includegraphics[width=\textwidth]{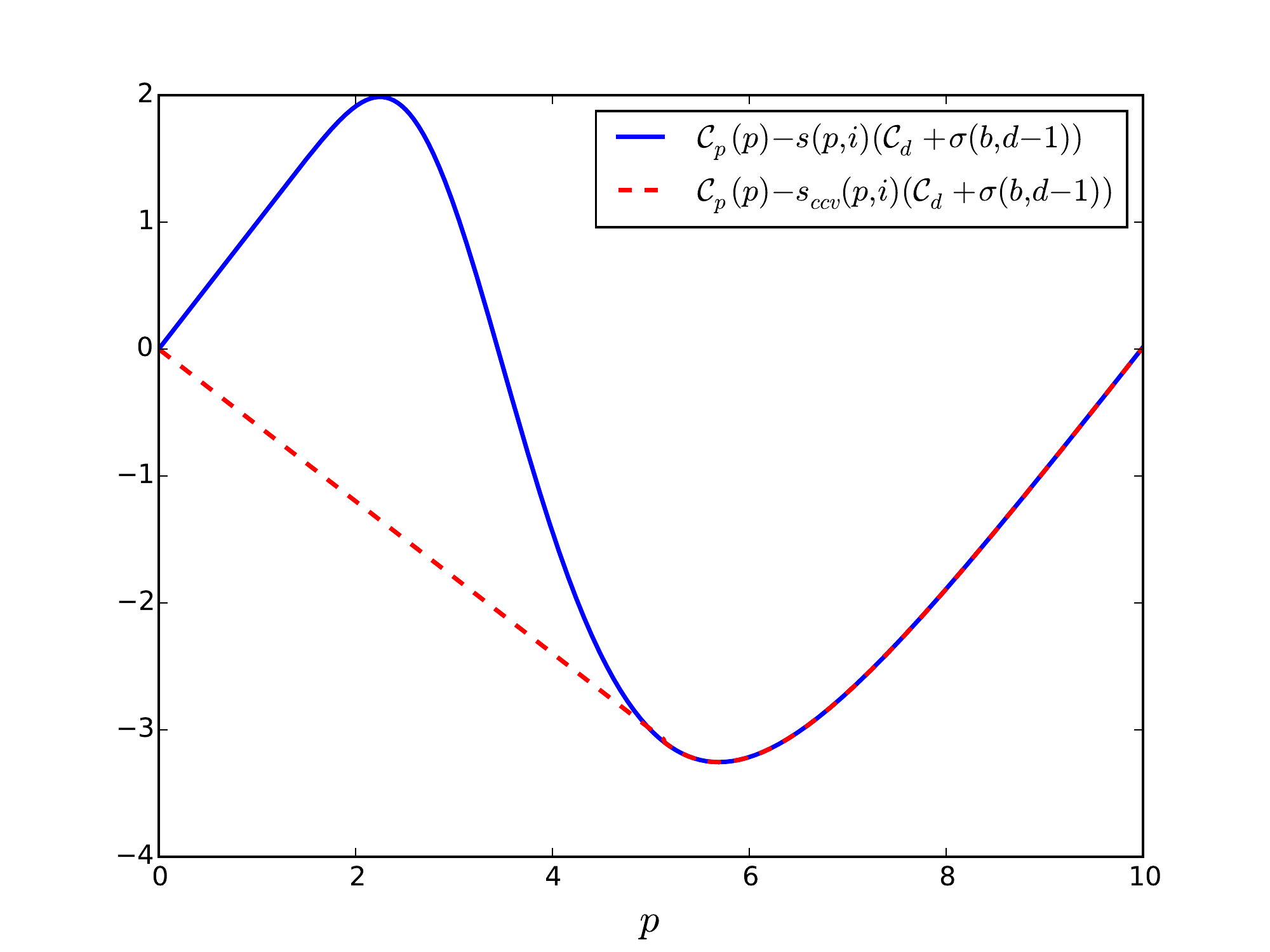}
    \caption{ 
      \label{fig:ccv-env-opt}}
  \end{subfigure}
  \caption{A depiction of how the concave envelope of a sigmoidal
    success probability function can be used to simplify $\gamma(b,
    d)$. Figure~\ref{fig:ccv-env} demonstrates how $s_{ccv}(\cdot, i)$
    compares to $s(\cdot, i)$. Figure~\ref{fig:ccv-env-opt} shows how
    using $s_{ccv}(\cdot, i)$ makes the first-order optimality
    condition sufficient for minimization.\label{fig:envelope}}
\end{figure}

\section{Sublinear-Backlog Power Control Algorithm \label{sec:SLBPC}}
Now we would like to to leverage the theoretical results above in
order to build a low-complexity power control scheme which
approximates the original optimal control scheme. Intuitively, these
low-complexity power controls are schemes which aim to mimic the
behavior of $\mu(b, d, i)$. In Section~\ref{ssec:control}, we defined
$\Pi$ as the set of deterministic policies but we will see that it can
be useful to use a stochastic scheme to capture the prominent
structural features proven above. To demonstrate our methodology, we
focus on a particular functional form of $s(p, i)$. We consider a
linear power cost and present two schemes. The first makes use of the
fact that $\mu(b,d)$ grows sublinearly with $\Cb(b)$ for any fixed
$d$. The second introduces a further refinement which incorporates how
$\mu(b, d)$ varies with $d$ and $i$.

\subsection{A Particular Success Probability Function}
The function $s(\cdot, \cdot)$ characterizes the wireless channel. In
Section~\ref{ssec:channel}, we gave some natural conditions that
$s(\cdot, \cdot)$ must satisfy but there is still considerable
flexibility available for modeling different types of channels. As
discussed above with regard to Propsition~\ref{prop:approx}, our
results apply to concave functions like $s(p, i) = p/(p + i)$ and
$s(p, i) = 1 - \exp(p/i)$ and also sigmoidal functions of the form
$s(p, i) = \left(1 - \frac{1}{2}\exp\left(-\beta_0 \frac{p}{i} +
    \beta_1\right)\right)^{\beta_2}$. To demonstrate the applicability
of our results, we choose to focus on a simple form. The same approach
can be applied with more complicated functions, but the closed forms
involved will be more cumbersome.

With this in mind, consider $s(p, i) = 1 - \exp(-p/i)$. Using the
notation from Proposition~\ref{prop:approx} gives us that $\phi(p, i) =
i\exp(p/i)$. When $T_b(0) \geq 0$ and $i = 1$,
\begin{align}
  \gamma(b, d) &= \ln\left(\frac{\Cd + \sigma(b,d-1)}{K}\right),\\
  \gamma(b, 2) 
  &= \ln\left(\frac{\Cd + \sigma(b,1)}{K}\right)
  = \ln\left(
    \frac{\Cd + \Cb(b) +\min_{p \in \Pcal}\{\Cp(p) - s(p,1)\Cd\}}
    {K}\right).
\end{align}
This demonstrates that the power (roughly) grows logarithmically with
the backlog cost and thus justifies the ``sublinear-backlog''
heuristic that will be utilized in the following schemes.

\subsection{Sublinear-Backlog Power Control 1 (SLBPC1)}
We first create a power control scheme adapts to the packet backlog,
but does not react to approaching deadlines or to fluctuating
interference. This gives us a simple policy which is a function of
only $b$. The previous calculation suggests that using $\gamma(b, 2)$
would be a reasonable approximate policy that meets this
requirement. Since $\gamma(b, 2)$ takes values in $\Rbb_{\geq 0}$
rather than $\Pcal$, we define $\hat\gamma$ as follows:
\begin{equation}
  \hat\gamma(b) = \min\argmin_{p \in \Pcal}\abs{p - \gamma(b, 2)}
\end{equation}
So SLBPC1 is a deterministic policy defined by $\hat\gamma(b)$. This
type of sub-optimal policy is useful because it does not require the
transmitter to know anything about the current
interference. Furthermore, $\hat\gamma$ has a simple form which can be
easily implemented. On the other hand, SLBPC1 does not adapt to
approaching deadlines or to fluctuating interference so it will not
behave entirely like the optimal control.

\subsection{Sublinear-Backlog Power Control 2 (SLBPC2)}
Our second power control scheme builds on SLBPC1 by incorporating our
understanding of how to react to interference and approaching
deadlines. Recall that for fixed interference, the sign of $T_b(0)$
determines whether the optimal policy will increase or decrease power
as the deadline approaches. We will now apply this idea to the case of
fluctuating interference. Let $f_b(i)$ as follows:
\begin{equation}
  f_b(i) = T_b(0) = \Cb(b) + \min_{p \in \Pcal}\left\{\Cp(p) - s(p, i)\Cd\right\}
\end{equation}
Define the ``signum'' function as $\sgn(x) = 2\Ind{x \geq 0} - 1$. We
will use $\sgn(f_b(i))$ to stochastically drift the power value up or
down as the interference and backlog change. More concretely, LBCP2 is
described as follows:
\begin{enumerate}
\item When there are $b$ packets in the buffer, the initial power level for a
  new HOL packet is $\hat\gamma(b)$.
\item If the previous transmission attempt is not successful, the transmitter
  computes $\sgn(f_b(i))$. If $\sgn(f_b(i))$ is positive, the power is increased
  to the next highest level with probability $p_{change}$. If $\sgn(f_b(i))$ is
  negative, the power is decreased to the next lowest level with probability
  $p_{change}$. If a higher/lower power level does not exist and the power needs
  to be increased/decreased, the power level remains the same.
\item This process is followed until the packet is successfully transmitted or
  until the deadline is exceeded. The process is then continued with a backlog
  of $b - 1$.
\end{enumerate}

We still need to choose a value for $p_{change}$. The packet will be
at the HOL for a number of time slots in $\{1, 2, \hdots, D\}$. As a
result, $p_{change}$ should be inversely proportional to $D$ so that
when $D$ becomes large $p_{change}$ becomes small. This corresponds to
the case of not having deadlines and the power control being constant
for the packet. In practice, a precise value can be chosen
experimentally but we take $p_{change} = 1/(2D)$.

SLBPC2 adapts to interference, deadlines, and backlog while SLBPC1 only adapts to
backlog. However, this requires an increase in complexity. SLBPC2 needs to store
the previous transmission power. SLBPC2 requires a reliable RSSI and also needs
to compute $\sgn(f_b(i))$. Finally, note that SLBPC2 requires a psuedo-random
number generator (PRNG). On the other hand, SLBPC2 is still low complexity when
compared to the optimal control. SLBPC2 does not require solving dynamic
programming equations and does not require knowledge of the interference
transition matrix $P$.

\section{Performance Evaluation \label{sec:performance}}
We take a two-pronged approach to the performance evaluation. In the
first approach, we verify that the SLBPC algorithms perform similarly
to the optimal control policy in some simple cases where the optimal
control can easily be computed. We then consider a more detailed
simulation in which power control is not on the same time scale as the
interference dynamics and there are also packet arrivals. In this more
detailed setting computing the optimal control is not tractable so we
compare SLBPC2 to some benchmark algorithms and show a statistically
significant increase in performance.

\subsection{Verification}
In this performance evaluation, we will compare the following power control
schemes:
\begin{itemize}
\item DP, the optimal control policy defined by $\mu(b, d, i)$. 
\item MIN, a benchmark power control scheme which always transmits at the
  minimum power.
\item MAX, a benchmark power control scheme which always transmits at the
  maximum power.
\item SLBPC1, the scheme defined above which does not adapt to deadlines or
  interference. 
\item SLBPC2, the scheme defined above which does adapt to approaching deadlines
  and fluctuating interference.
\end{itemize}
We will take $B = 20$, $D = 5$, $\Cb(b) = b$, $s(p, i) = 1 - \exp(p/(2i))$,
$\Pcal = \{0.1, 0.2, 0.4, 0.8\}$, and $\Cd = 1$. We will take $\Cp(p) = Kp$
where $K > 0$ will be varied to show how the power control schemes are affected
by different power sensitivities. Note that these parameters are not meant to
describe a specific system. This performance evaluation is meant to illustrate
the various trade-offs between the different power control schemes.

We assume the interference takes two possible values $i_t \in \{1,
2\}$. Let $p_u$ be the probability that the interference moves from
low to high (i.e. up) and let $p_d$ is the probability of the
interference moving from high to low (i.e. down). The transition
matrix is therefore given by
\begin{equation}
  P = 
  \left[
    \begin{array}{cc}
      1 - p_u & p_u\\
      p_d & 1 - p_d
    \end{array}
  \right].
\end{equation}
By varying $p_u$ and $p_d$, we can simulate slow fading and fast fading
channels. In particular, when $p_u$ and $p_d$ are small, we have a slow fading
channel and when $p_u$ and $p_d$ are large, we have a fast fading channel. It is
important to investigate these different cases because SLBPC1 and SLBPC2 were
designed with properties from the fixed interference case. Fixed interference
can be a good approximation of a slow fading channel but typically not of a fast
fading channel.

For each scheme we vary $K$, $p_u$, and $p_d$ and compute different
performance metrics via Monte Carlo simulation with 2000 particles. Because of
how we formulated the optimal control problem, the obvious metric is the total
cost. In addition, we compute the average fraction of dropped packets and the
average power per packet.

\subsubsection{Slow Fading}
We begin by simulating a slow fading channel in which $p_u = p_d = 0.1$. Because
$1/p_u = 1/p_d = 10$ and $D = 5$, we expect that each packet will only
experience a fixed level of interference. We first vary $K$ and plot the
estimated total cost incurred by each of the power control schemes. The results
are shown in Figure~\ref{fig:slowFadeCost}. When the power sensitivity is low
(i.e. small values of $K$), the DP and SLBPC schemes perform like the MAX scheme
and when the power sensitivity is high (i.e. large values of $K$), the DP and SLBPC
schemes perform like the MIN scheme. For moderate power sensitivities, the DP
outperforms both SLBPC schemes but the SLBPC cost curves are qualitatively similar
to the DP cost curve. We also see that when it comes to total cost, the SLBPC
schemes are almost indistinguishable.

\begin{figure}
  \centering
  \begin{subfigure}{0.5\textwidth}
    \includegraphics[width=\textwidth]{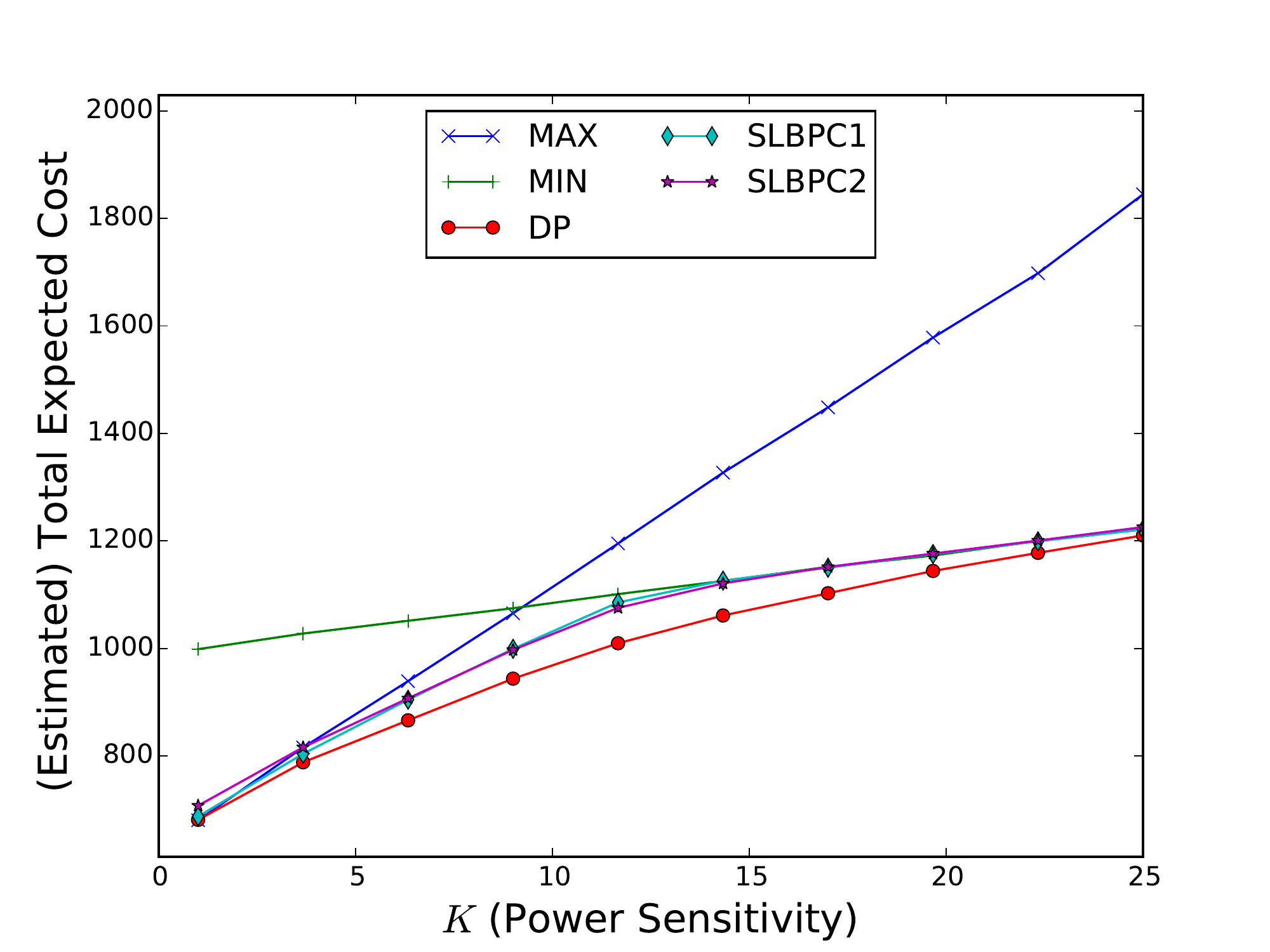}
    \caption{ Total Cost vs. Power Sensitivity, Slow Fading 
      \label{fig:slowFadeCost}}
  \end{subfigure}%
  \begin{subfigure}{0.5\textwidth}
    \includegraphics[width=\textwidth]{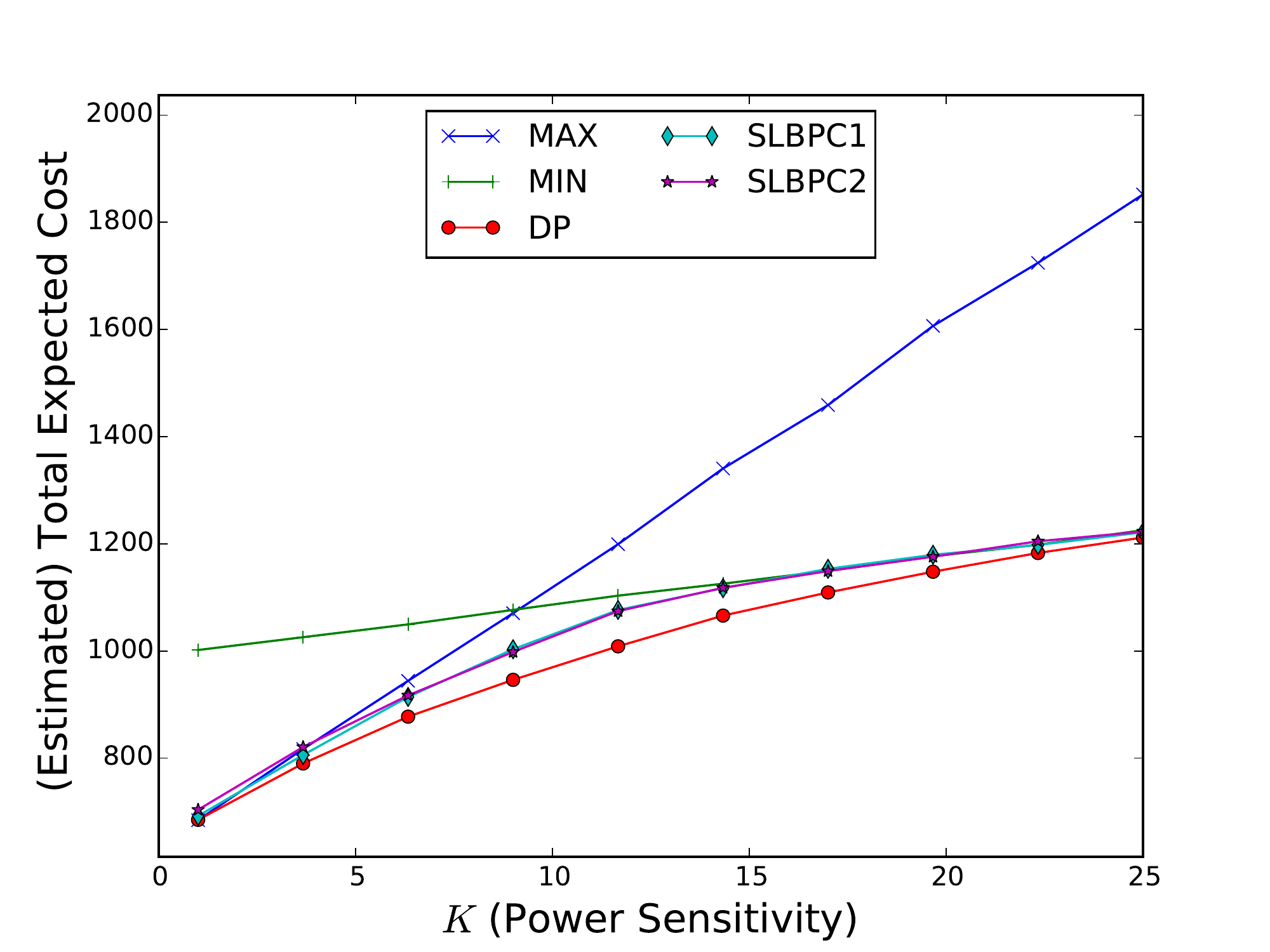}
    \caption{ Total Cost vs. Power Sensitivity, Fast Fading
      \label{fig:fastFadeCost}}
  \end{subfigure}
  \caption{
    Estimated total cost as a function of power sensitivity. The power cost is
    $\Cp(p) = Kp$ so a lower value of $K$ indicates less sensitivity of power
    while a higher value of $K$ indicates a greater sensitivity. In the slow
    fading case, $p_u = p_d = 0.1$. In the fast fading case, $p_u = p_d =
    0.9$. For low values of $K$, the DP and SLBPC schemes are just like the MAX
    scheme. For high values of $K$, the DP and SLBPC schemes are just like the
    MIN scheme. For moderate values of $K$, both SLBPC schemes qualitatively
    match the DP scheme. These observations hold for both the slow fading and
    fast fading cases.
    \label{fig:cost}
  }
\end{figure}

The SLBPC schemes achieve a very similar total cost, but we know that they do not
achieve this total cost in the same way. Because the SLBPC2 scheme adapts to
interference and deadlines, SLBPC2 behaves more similarly to the DP scheme than
SLBPC1 does. To demonstrate this, we choose a few values for $K$ and compute the
average power per packet and the average fraction of dropped packets. The
results are shown in Figure~\ref{fig:slowFadeMetrics}. In each case, SLBPC2 is
more similar to the DP scheme than SLBPC1 is. SLBPC1 and SLBPC2 may achieve similar
total costs, but SLBPC2 is more similar in behavior to the DP scheme.

\begin{figure}
  \centering
  \begin{subfigure}{\textwidth}
    \includegraphics[width=0.5\textwidth]
    {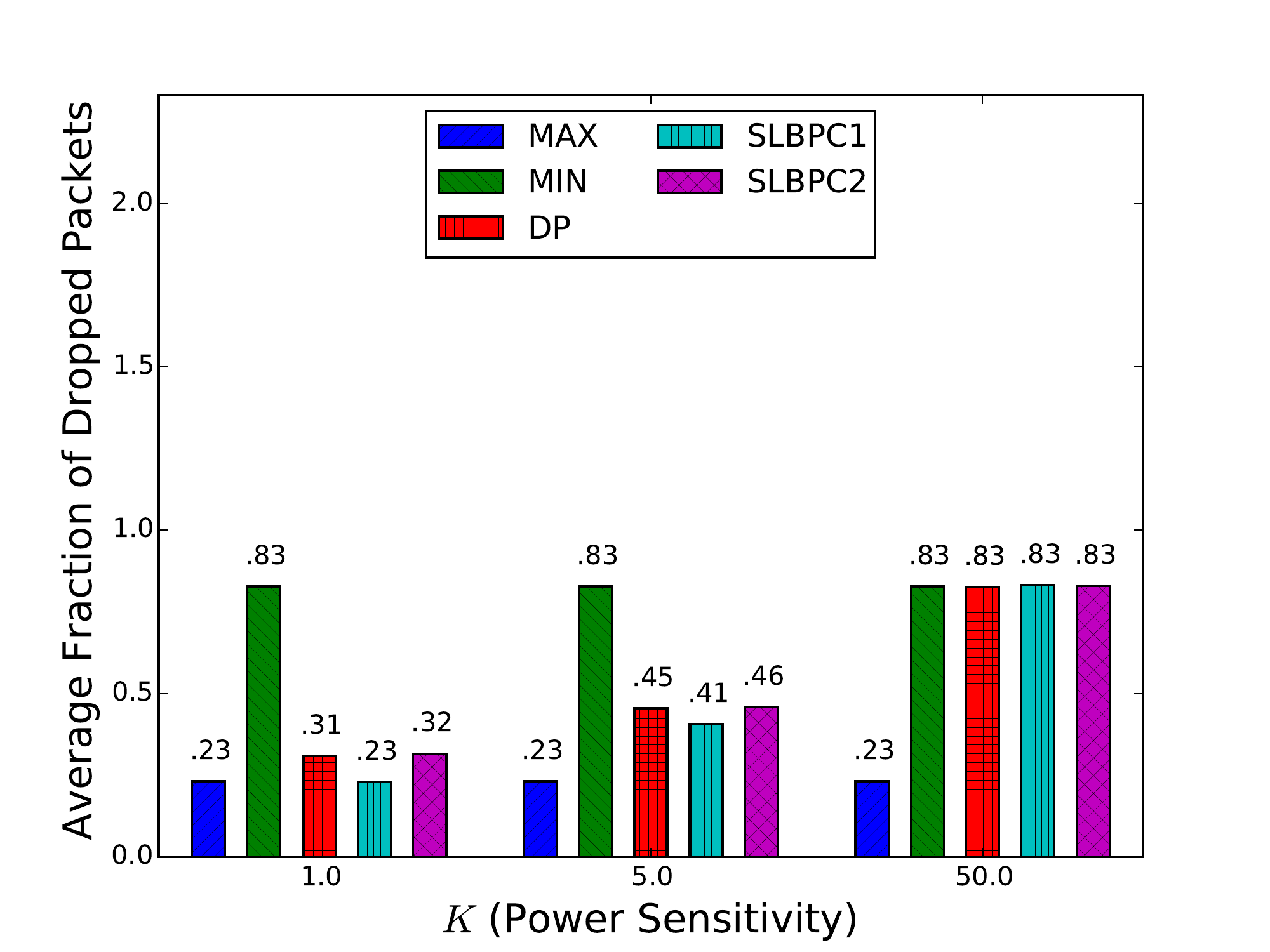}%
    \includegraphics[width=0.5\textwidth]
    {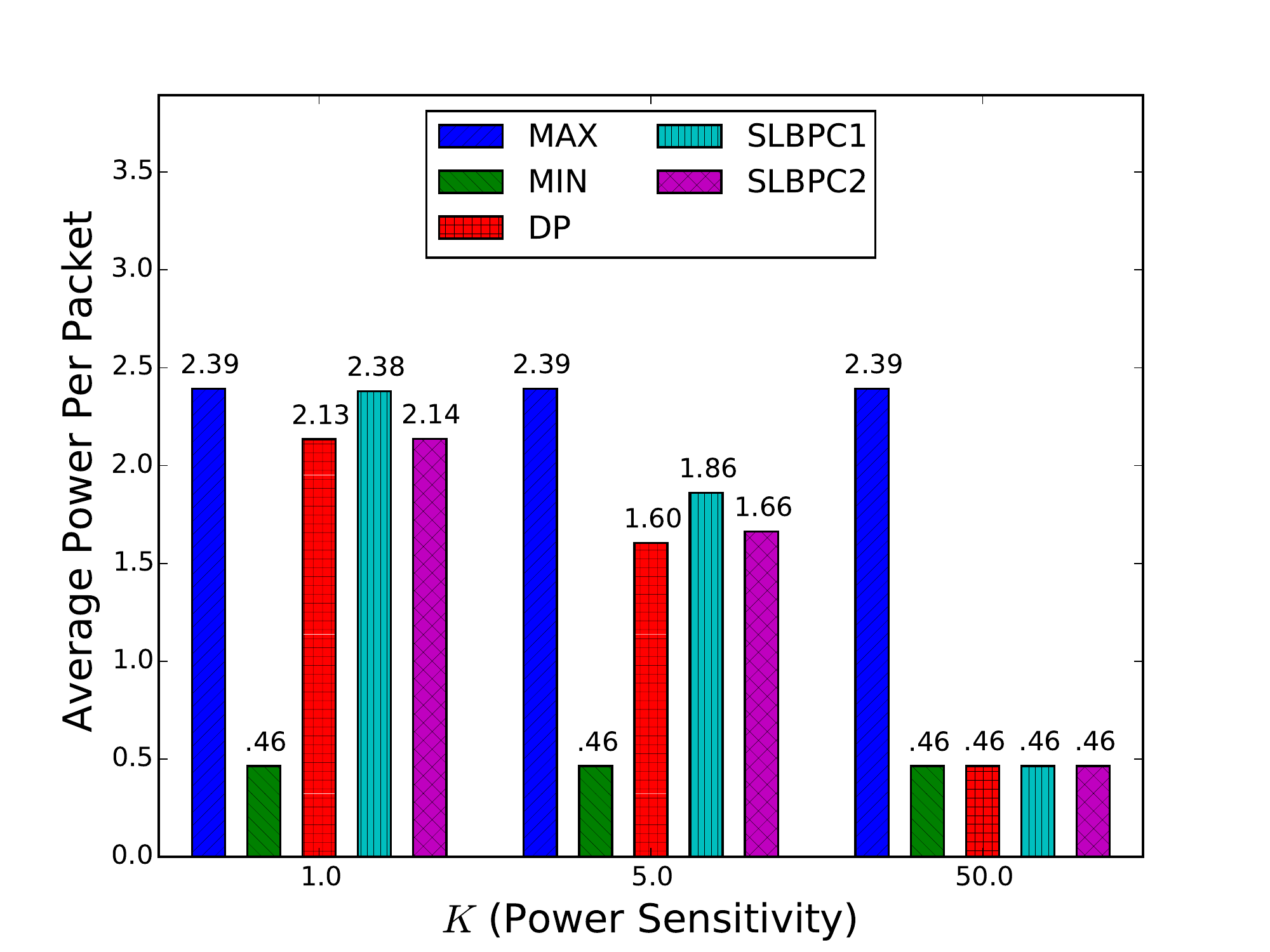}
    \caption{ Performance Metrics for Various $K$, Slow Fading
      \label{fig:slowFadeMetrics}}
  \end{subfigure}
  \begin{subfigure}{\textwidth}
    \includegraphics[width=0.5\textwidth]
    {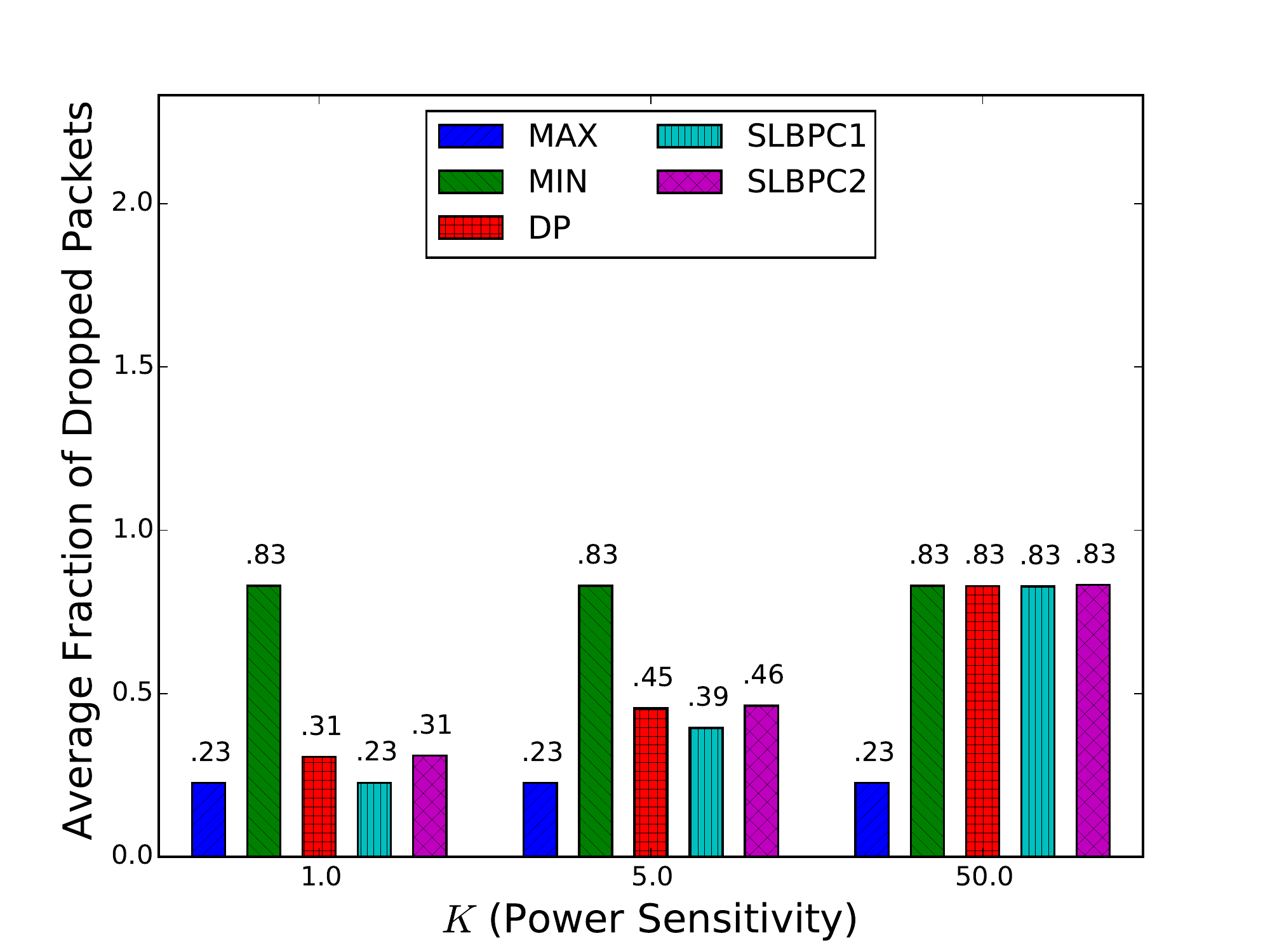}%
    \includegraphics[width=0.5\textwidth]
    {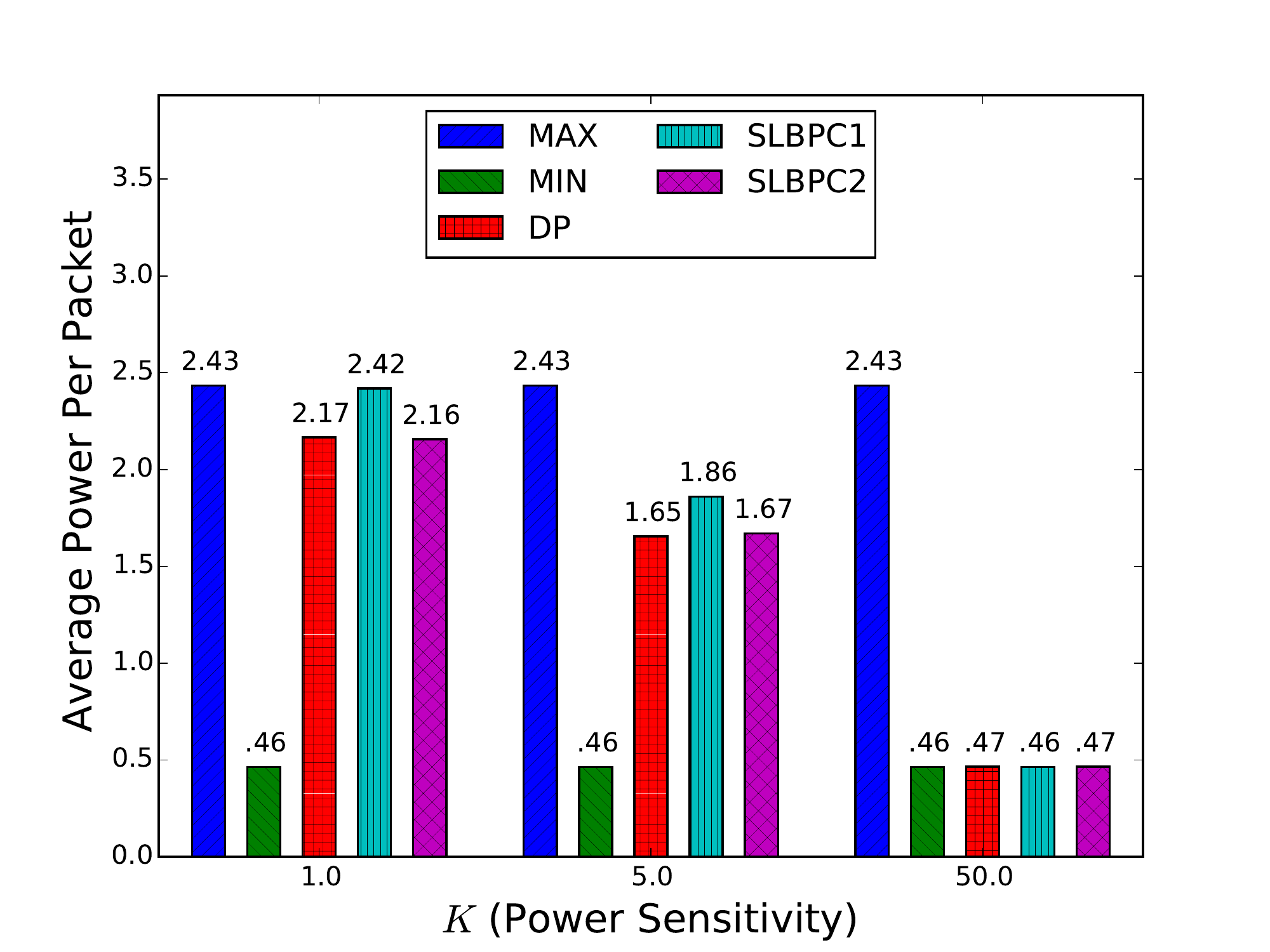}
    \caption{ Performance Metrics for Various $K$, Fast Fading
      \label{fig:fastFadeMetrics}}
  \end{subfigure}
  \caption{
    For a few different values of $K$, the schemes are compared on the average
    fraction of dropped packets and the average power expended per packet.  In
    the slow fading case, $p_u = p_d = 0.1$. In the fast fading case, $p_u = p_d
    = 0.9$. We see that SLBPC2 behaves more similarly to the DP scheme than SLBPC1
    even though the SLBPC schemes have similar total costs. This demonstrates
    that reacting to deadlines and adapting to fluctuating interference is
    important when approximating the optimal control.
    \label{fig:metrics}
  }
\end{figure}

\subsubsection{Fast Fading}
Now we consider a fast fading channel for which $p_u = p_d = 0.9$. Because
$1/p_u = 1/p_d = 1.\overline{11}$ while $D = 5$, we expect the interference to
change multiple times during the transmission of a single packet. As a result,
having knowledge of the interference transition matrix $P$ could give the DP
power control scheme a significant advantage over the SLBPC schemes. 

However, our simulations show that the SLBPC schemes do not degrade in
performance due to a fast fading channel. In Figure
\ref{fig:fastFadeCost}, we plot the total cost as a function of the
power sensitivity $K$. As with the slow fading channel, we see that
the DP scheme acts like the MAX scheme for low $K$, like the MIN
scheme for high $K$, and is better than all schemes for moderate
$K$. We also see that the cost curves for the SLBPC schemes are
qualitatively similar to the cost curve for the DP scheme. In
addition, the optimality gap is not significantly greater than it was
in the case of a slow fading channel. Furthermore, when considering
the average fraction of dropped packets and the average power per
packet in Figure~\ref{fig:fastFadeMetrics}, we see the same trends as
before. Both SLBPC schemes are similar to the DP scheme in terms of
total cost, but SLBPC2 is more similar to the DP scheme in terms of
behavior.

\subsection{Detailed Simulation}
In this section we consider a more realistic simulation in which our
assumptions do not always hold. We fix $\Cd = 1$, $\Cb(b) = b$, $D =
3$, and $\Pcal = \{0.5, 1.0, 2.0, 4.0, 8.0\}$.  The transmitter starts
with $B = 100$ packets and in each time slot a new packet arrives with
probability $0.1$. We consider a two-state channel where the low
interference level is 8.0 and the high interference level is 16.0. The
transition matrix for the channel is given by
\begin{equation}
  P = 
  \left[
    \begin{array}{cc}
      0.1 & 0.9\\
      0.9 & 0.1
    \end{array}
  \right]
\end{equation}
so that we have a fast fading channel. In addition, we assume that the
channel dynamics are at a time scale that is twice as fast as the
transmission dynamics. In other words, the interference will take two
values over the course of a transmission time slot and the transmitter
only has knowledge of the first value. If the transmitter uses power
level $p$ and the interference is at level $i$ and then $i'$, the
probability of successfully transmitting the packet is
$$s_{chan}(p, i, i') = 1 - \exp(-p / \max(i, i')).$$
Since the transmitter only knows the value of $i$ (and not $i'$), the
SLBPC algorithms use
$$s(p, i) = 1 - \exp(-p / i)$$
as a success probability function. 

We compare the SLBPC algorithms to a power control policy call AVG
which is parameterized by $\alpha \in [0, 1]$. Let $p_{min} =
\min\Pcal$ and $p_{max} = \max\Pcal$. Then with probability $\alpha$,
AVG uses $p_{min}$ and with probability $1 - \alpha$, AVG uses
$p_{max}$. By varying $\alpha$, AVG can (on average) use any power
level between $p_{min}$ and $p_{max}$. For example, when $\alpha = 1$,
AVG is equivalent to MIN and when $\alpha = 0$, AVG is equivalent to
MAX.

We simulate the performance of SLBPC1, SLBPC2, and AVG to calculate
the average packet error rate and the average total time until the
transmitter has zero remaining packets. We vary $K$ (for SLBPC) and
$\alpha$ (for AVG) to get curves that demonstrate the trade-offs
between power, packet error rate, and total transmission time. We run
$10^4$ simulations for each point on the curve. The results are shown
in Figure~\ref{fig:curves}. Note that the power insensitive regime
(i.e. low $\alpha$ and low $K$) is in the lower right and the power
sensitive regime (i.e. high $\alpha$ and high $K$) is in the upper
left. In the power insensitive regime, all power control schemes will
be roughly equivalent because they will all use high power levels for
transmission. Note that the optimal control framework did not
explicitly optimize for minimizing packet drop rate or for total
transmission time. The costs used implicitly account for these
performance metrics but it isn't immediately clear how SLBPC1 and
SLBPC2 will compare.

\begin{figure}
  \centering
  \begin{subfigure}{0.5\textwidth}
    \includegraphics[width=\textwidth]{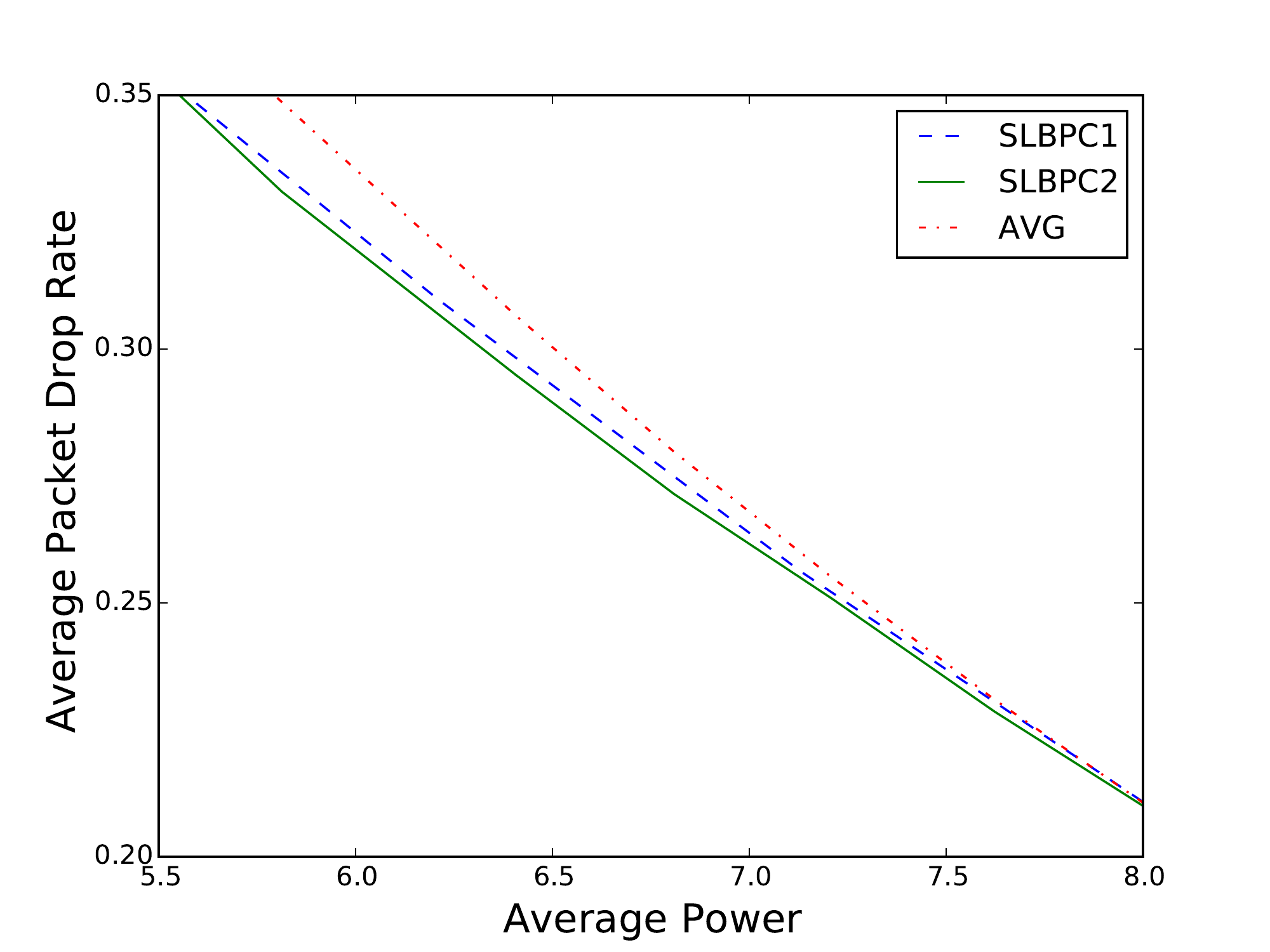}
    \caption{Packet Drop Rate vs. Power\label{fig:pow_err}}
  \end{subfigure}%
  \begin{subfigure}{0.5\textwidth}
    \includegraphics[width=\textwidth]{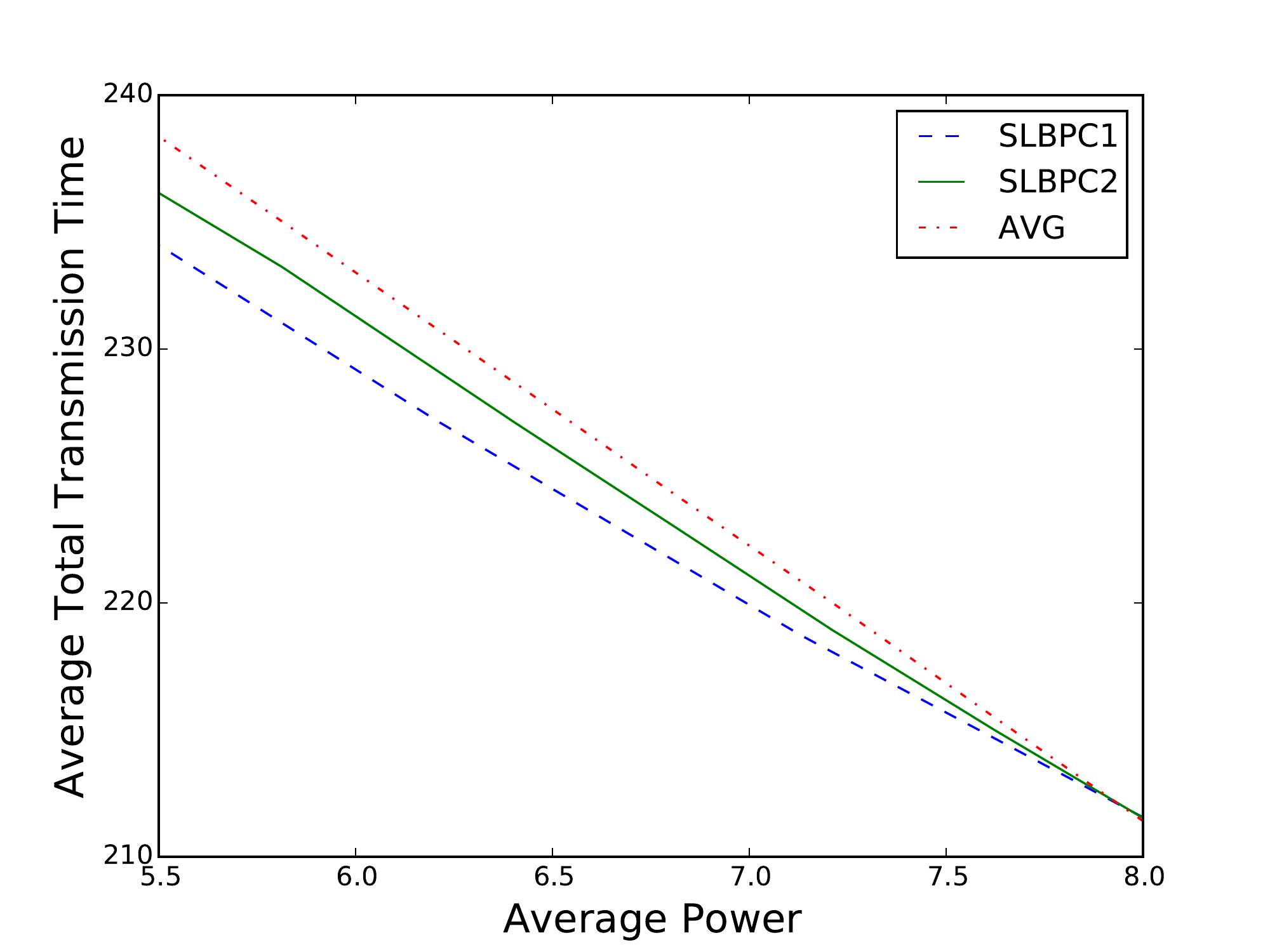}
    \caption{Total Transmission Time vs. Power\label{fig:pow_time}}
  \end{subfigure}
  \caption{ Trade-off curves comparing SLBPC1, SLBPC2, and AVG. 
    \label{fig:curves}
  }
\end{figure}

We see that for a given power level, SLBPC1 achieves a lower packet
drop rate than AVG and that SLBPC2 achieves a lower packet drop rate
than SLBPC1. The advantage of the SLBPC algorithms over AVG is most
apparent in the low power regime. This demonstrates that the SLBPC
algorithms enable a higher fidelity packet stream with less power
used. The SLBPC algorithms also reduce the total amount of
transmission time. However, it is interesting to note that SLBPC1 is
better at reducing the total transmission time than SLBPC2 is. Recall
that SLBPC1 is a function only of the backlog so SLBPC1 makes reducing
backlog pressure the primary objective. Since backlog pressure is a
proxy for total delay, this behavior is expected.

\section{Conclusions\label{sec:conclusions}}
This paper examines the problem of transmitting packets across a
stochastically fluctuating wireless link while balancing power usage,
overall latency, and jitter constraints. By considering a special case
of the optimal transmitter power control problem, we have developed
mathematically justified heuristics that are useful when designing
sub-optimal power control schemes. In particular, we have analytically
characterized how the optimal power control should vary with the
backlog and how the optimal power control should react to approaching
deadlines. By incorporating these structural properties into the
Sublinear-Backlog Power Control (SLBPC) scheme, we have demonstrated
how our theoretical results can be leveraged to build low-complexity
approximations of the optimal power control. Monte Carlo simulations
show that the SLBPC scheme is a good approximation for the optimal
power control scheme. An additional simulations show that SLBPC
outperforms benchmark algorithms in realistic settings.

\section*{References}
\bibliographystyle{elsarticle-num}
\bibliography{NMaster_NBambos_PowerEfficientStreaming}

\end{document}